\documentclass{article}

\usepackage[utf8]{inputenc}
\usepackage{amsmath, amssymb, amsthm, authblk, mathtools, dsfont, float, bbm, hyperref, caption, adjustbox, kantlipsum, algorithm, algcompatible, xcolor, comment, graphicx, cleveref, booktabs, paralist, subcaption, fullpage}

\usepackage[round]{natbib}

\newtheorem{theorem}{Theorem}

\newtheorem{lemma}{Lemma}
\theoremstyle{definition}
\newtheorem{definition}{Definition}
\newtheorem{assumption}{Assumption}

\newcommand{\bbR}{\mathbb{R}}
\newcommand{\bbE}{\mathbb{E}}
\newcommand{\bbP}{\mathbb{P}}
\newcommand{\Var}{\text{Var}}
\newcommand{\st}{\text{ } \vert \text{ } }
\newcommand{\indic}{\mathds{1}}
\newcommand{\indep}{\perp \!\!\! \perp}
\newcommand{\nindep}{\not\!\perp\!\!\!\perp}

\newcommand{\at}{\mathrm{AT}}
\newcommand{\nt}{\mathrm{NT}}
\newcommand{\co}{\mathrm{C}}
\newcommand{\oracle}{\mathrm{oracle}}
\newcommand{\plugin}{\mathrm{plug-in}}
\newcommand{\op}{\mathrm{op}}

\newcommand{\simind}{\stackrel{\mathrm{ind}}{\sim}}
\newcommand{\dbern}{\mathrm{Bern}}

\newcommand{\diag}{\mathrm{diag}}
\newcommand{\late}{\mathrm{LATE}}

\newcommand{\rd}{\,\mathrm{d}}
\newcommand{\phe}{\phantom{=}}

\title{Constrained Design of a Binary Instrument in a Partially Linear Model}
\author[1]{Tim Morrison}
\author[2]{Minh Nguyen}
\author[2]{Jonathan Chen}
\author[3,1]{Michael Baiocchi}
\author[1]{Art B. Owen}
\affil[1]{Department of Statistics, Stanford University}
\affil[2]{Department of Biomedical Data Science, Stanford University}
\affil[3]{Department of Epidemiology and Population Health, Stanford University}
\date{May 2025}

\begin{document}

\maketitle

\begin{abstract}
We study the question of how best to assign an encouragement in a randomized encouragement study. In our setting, units arrive with covariates, receive a nudge toward treatment or control, acquire one of those statuses in a way that need not align with the nudge, and finally have a response observed. The nudge can be modeled as a binary instrument if one assumes that it affects the response only via the treatment status. Our goal is to assign the nudge as a function of covariates in a way that best estimates the local average treatment effect (LATE). We assume a partially linear model, wherein the baseline model is non-parametric and the treatment term is linear in the covariates. Under this model, we outline a two-stage procedure to consistently and optimally estimate the LATE. Though the variance of the LATE is intractable, we derive a finite sample approximation and thus a design criterion to minimize. This criterion is convex, allowing for constraints that might arise for budgetary or ethical reasons. We prove conditions under which our solution asymptotically recovers the lowest true variance among all possible nudge propensities. A one-stage version of the algorithm is consistent but not necessarily optimal. We apply our method to a semi-synthetic example involving triage in an emergency department and find significant gains relative to a regression discontinuity design.
\end{abstract}

\section{Introduction}
In this paper, we study the causal effect of a binary treatment
on subjects who do not necessarily comply with
their treatment assignment.  For a heterogeneous population with covariates measured prior
to treatment, it can be beneficial to vary treatment probabilities using methods from experimental design to optimally estimate the causal effect of the treatment. 
However, the potential for noncompliance presents difficulties for this design problem because the experimenter cannot directly choose these probabilities.

This problem arises in several practical settings.  For instance, suppose that an electronic commerce company has introduced a new AI shopping assistant on their website and is interested in the average effect of opting into this assistant on downstream purchases. It may be the case that, independently of the assistant's effect itself, users who opt into allowing it are also more likely to purchase due to some unmeasured confounder such as being ``tech savvy." This creates a dependence between treatment status and potential outcomes even conditional on observed covariates that renders standard causal inference methods invalid. 
Another example, which we discuss in this paper in greater detail, concerns triage in an emergency department (ED). There, patients who are sent to the intensive care unit (ICU) may differ in unmeasured ways from those who are not, making it difficult to estimate the causal effect of being sent to the ICU.

To handle noncompliance issues, researchers often use a randomized encouragement design \citep{zelen1979, holland1988, westRCTalts, bradlow, imai2013}, in which subjects are randomly nudged toward or away from treatment.  Under mild conditions, this nudge is a binary instrumental variable (IV), enabling IV methods to be used for analysis. The goal in the randomized encouragement setting then changes from finding optimal treatment probabilities to finding optimal encouragement probabilities for all subjects.

Because of ethical and economic considerations that arise in real-world problems, we are particularly interested in being able to derive optimal designs of this nudge under constraints. For instance, there may be a budget constraint on the expected number of treated units or an ethical constraint that no more deserving subject has a lower encouragement probability than some less deserving subject. We may also want to impose covariate balance constraints, as described in \cite{morrisonowen}. 

We study efficiency using a partially linear model, with a non-parametric baseline mean response and a linear model for the treatment effect. 
This model sacrifices some flexibility relative to fully non-parametric IV work based on deep learning, such as \cite{hartford, chernozhukov2017, bennett2019, syrgkanis2019}. However, our view is that the partially linear model retains desirable flexibility in the baseline model while giving a concrete and practical criterion for the design stage. A linear treatment effect is also more readily interpretable than a black box \citep{moosavi2023costs}. 

By adapting some techniques for the partially linear model from \cite{robinson}, we derive a method to estimate the local average treatment effect (LATE) in our setting. This is the average treatment effect on those who comply with the nudge and is estimable without assuming unconfoundedness. Crucially, the variance of our estimator can be approximated by a quantity that is convex in the encouragement probabilities, which are what the experimenter is allowed to choose. This enables the construction of a convex optimization procedure and a notion of the optimal encouragement under convex constraints. We show as well that, under certain assumptions presented in Section \ref{sec:methodology}, these approximations to the true variance make a negligible difference
as the sample size tends to infinity. 

We include a semi-synthetic numerical example motivated by the ED triage problem discussed earlier.  Our example makes use of
a machine learning-derived prediction of whether a newly admitted patient is likely to end up in the ICU eventually.  When that estimated probability is very high, it is natural that we might recommend admission to the ICU more often. There are several outcomes of interest in this setting, including the length of hospital stay.  More details about this motivating example are given in Section~\ref{sec:triage}.

The problem of optimally choosing a propensity score to estimate a treatment effect of interest has been studied in a variety of settings. These include under a linear model with design constraints \citep{liowenTBD, morrisonowen}, for batch-adaptive designs \citep{hahnbatch, liowenbatch}, and for sequential designs \citep{vanderlaan2008, kato2021}. However, this question remains minimally explored in the context of confounding and instrumental variables. Several papers, e.g., \cite{kuersteiner, kuang2020}, consider the problem of optimally combining a set list of candidate instruments, which is conceptually similar to our question of how best to design a single instrument. 

Most similar to our work is the recent paper of \cite{chandak2023adaptive}, who study how to choose an optimal policy for instrumental variable assignment in some policy class. Their approach combines the Deep IV framework of \cite{hartford} with the use of influence functions to obtain an estimate of the gradient of the loss with respect to the policy parameters. While addressing the same problem, their method does not yield a convex objective and so has no guarantee of finding the global optimum. In addition, they do not consider the case in which there are constraints on the design, which is important for our motivating example. 

This paper is organized as follows. In Section \ref{sec:setup}, we introduce the key notation, assumptions, and parameters of interest. In Section \ref{sec:methodology}, we derive our proposed method and state some relevant theoretical guarantees. In Section \ref{sec:triage}, we apply our method to a semi-synthetic example based on an emergency department dataset. We present additional proofs in Appendix~\ref{sec:proofs}, and we discuss generalizations of our procedure involving cross-fitting and heteroscedasticity in Appendices \ref{sec:cross-fitting} and \ref{sec:WLS}, respectively. We provide further details of our simulation in Appendix~\ref{sec:expdetails}. 

\section{Setup} \label{sec:setup}
We use $X_i \in \mathcal{X} \subseteq \mathbb{R}^d$ to denote the covariates of unit $i$ and $W_i \in \{0, 1\}$ to denote their (actual) treatment status. Here, $W_i = 0$ and $W_i = 1$ correspond respectively to control and treatment. We let $Y_i \in \bbR$ denote the response of interest for unit $i$. We adopt the potential outcomes framework \citep{rubin}. We also make the stable unit treatment value assumption (SUTVA) that a subject's response depends only on their own treatment assignment. 

Importantly, however, we do not make the standard unconfoundedness assumption that a unit's treatment status is conditionally independent of their response given $X$. Concretely, if we write 
\begin{align} \label{generalequation} Y = g(X, W) + \varepsilon, \end{align} 
where $\bbE[\varepsilon] = 0$, then as in \cite{hartford}, we allow for $\bbE[\varepsilon \st X, W] \neq 0$.

Instead, we assume access to a binary instrumental variable $Z \in \{0, 1\}$, which only impacts $Y$ through its effect on the treatment status $W$. This $Z$ is the binary ``nudge'' that we mention in the introduction. The nudges $Z_i$ are independent but not identically distributed Bernoulli random variables. Such instruments arise by construction in randomized encouragement designs, in which the experimenter can only encourage a subject to have a certain treatment status but cannot choose it directly. 

We write $W_i(Z_i)$ for $Z_i \in \{0, 1\}$ to denote the treatment status that subject $i$ would have if given the nudge status $Z_i$. Similarly, we write $Y_i(W_i)$ for the potential outcome for subject $i$ corresponding to the treatment $W_i$. We could write $Y_i(Z_i, W_i)$ to account for both the nudge status and the treatment status, but we will assume below that this is not necessary. 

We are now ready to formally state the conditions on our instrument $Z$. 
\begin{definition}\label{def:IV} An instrumental variable $Z \in \{0, 1\}$ is any random variable that satisfies: 
\begin{enumerate}
\item Exclusion restriction: $Y(w, z) = Y(w)$ for all $(w, z) \in \{0, 1\}^2$.
\item Conditional exogeneity: $Z \indep \{Y(0), Y(1), W(0), W(1)\} \st X$.
\item Relevance: $Z \nindep W \st X$. 
\end{enumerate} 
\end{definition}
\noindent
Of note, our assumptions on the instrument $Z$ are conditional on $X$, as in \cite{abadie}. This is because, in our motivating setting, we allow the distribution of $Z$ to depend on $X$. 

In our setting, $Z$ is generated randomly from a distribution depending only on $X$, and so conditional exogeneity holds. For the applications we have in mind, it is not plausible that the nudge would be completely ignored in the choice of treatment $W$, and so relevance is a very reasonable assumption. This leaves the exclusion restriction, which states that the outcome $Y$ only depends on the nudge through the nudge's effect on the treatment. While often hard to justify in observational settings, the example of a nudge in an encouragement design is more compelling, since a nudge is unlikely to impact the outcome via other plausible mechanisms. 

Because the instrument is binary, our analysis will hinge on a subject's \textit{compliance class}, as in \cite{AIR1996, abadie}. We define a subject to be a \textit{complier} (C) if $W(0) = 0$ and $W(1) = 1$, an \textit{always-taker} (AT) if $W(0) = W(1) = 1$, a \textit{never-taker} (NT) if $W(0) = W(1) = 0$, and a \textit{defier} (D) if $W(0) = 1$ and $W(1) = 0$. 

For our motivating example, a subject corresponds to an admission to a hospital's ED. This includes the patient and any of their covariates, such as their medical history and the physician looking after them.  An always-taker would then be a case severe enough that the physician sends the patient to an ICU regardless of the nudge status. Analogously, a never-taker is likely a mild case, and a complier is likely an intermediate case. 

The compliance status of an observation depends on $X$ but, in our model, is also random.
We define the compliance probability 
$p_{\co}(X) = \bbP(\text{complier} \st X)$, with $p_{\at}(X)$, $p_{\nt}(X)$ and $p_{\mathrm{D}}(X)$
defined analogously for always-takers, never-takers and defiers, respectively. The relevance criterion in Definition~\ref{def:IV} can only be satisfied if $p_{\co}(X)+p_{\mathrm{D}}(X)>0$ at all $X$. 

Defiers correspond to units who actively do the opposite of whatever they are encouraged to do. For example, a user of a tech platform may be dissuaded from opting into a new feature explicitly because of their annoyance at a push notification but would have done so if uncontacted. However, this is often implausible, and so we will assume that there are no defiers. This assumption, which will help with subsequent identifiability results, is often known as monotonicity and is common in the IV literature \citep{imbensangrist1994}. 

\begin{assumption}\label{as:nodefiers} (No defiers) 
For all $X$, $p_{\mathrm{D}}(X)=0$.
\end{assumption}

Note that in some settings, treatment is only open to those receiving the nudge, in which case there are no defiers. This is sometimes called strong treatment access monotonicity \citep{jin2008}.  

We will also need to assume that there are some compliers at every $X$. While this follows from Assumption~\ref{as:nodefiers} and the relevance condition in Definition~\ref{def:IV}, we list it as an assumption in its own right.
\begin{assumption}\label{as:somecompliers} (Some compliers)
For all $X$, $p_{\co}(X)>0$.
\end{assumption}

When there is noncompliance, the average treatment effect (ATE), 
$$\tau = \bbE[Y(1) - Y(0)],$$ 
is no longer identified. Instead, a standard target of interest \citep{imbensangrist1994, AIR1996} is the local average treatment effect (LATE):
$$\tau_{\late} = \bbE[Y(1) - Y(0) \st \text{complier}].$$
The target $\tau_{\late}$ corresponds to the average effect of treatment on those subjects who do respond to the nudge $Z$. Under our assumptions on $Z$, $\tau_{\late}$ can be written as
\begin{align} \label{LATEformula} \tau_{\late} = \bbE_{X \st \co}\left[\frac{\bbE[Y \st X, Z = 1] - \bbE[Y \st X, Z = 0]}{\bbE[W \st X, Z = 1] - \bbE[W \st X, Z = 0]} \right].\end{align}
Here, $X \st \co$ refers to the distribution of $X$ conditional on being a complier. This result is a consequence of Theorem 1 of \cite{frolich2007}, though we reprove it in Appendix \ref{ssec:proofLATE}, for completeness. 

\begin{figure}[t!] 
\centering
\includegraphics[width=8cm]{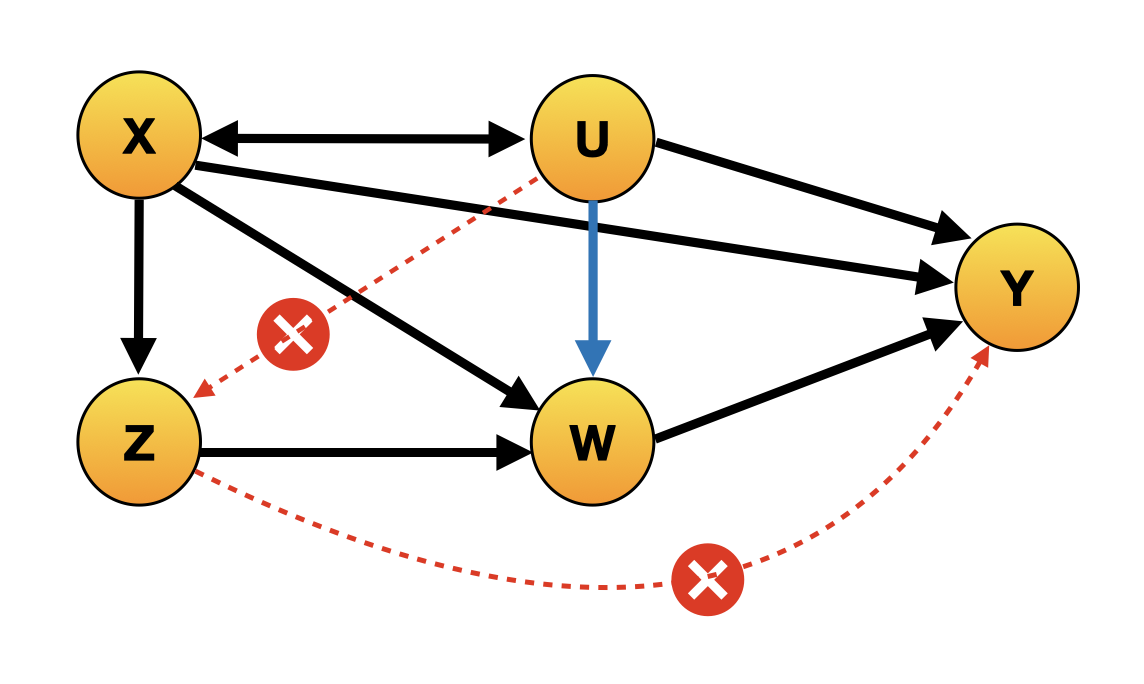}
\caption{Illustration of the directed acyclic graph for our problem. Terms included are: the observed covariates $X$, the unmeasured confounders $U$, the binary instrument $Z$, the treatment status $W$, and the response $Y$. The covariates $X$ and the confounders $U$ can be related in arbitrary ways. The treatment $W$ may be affected by the confounders $U$ (blue arrow), but the binary instrument $Z$ affects the outcome $Y$ only via its effect on $W$. The dashed red arrows indicate forbidden relationships involving $Z$. Figure inspired by a similar diagram from \cite{chandak2023adaptive}.}
\label{fig:dag}
\end{figure}

Of note, the denominator of \eqref{LATEformula} is precisely $p_{\co}(X)$, since the first expectation in the denominator accounts for both compliers and always-takers whereas the second term accounts for only always-takers.

In this paper, we study the effect of the nudge distribution $Z \st X$ on efficient estimation of $\tau_{\late}$. We consider a scenario in which subjects arrive with known covariates $X_i \in \mathcal{X} \subseteq \bbR^d$ and unknown compliance status $C_i \in \{\text{AT}, \text{NT}, \text{C}\}$ drawn IID from some distribution $P_{X \times C}$. Subjects are given a nudge $Z_i \in \{0, 1\}$, then receive a treatment $W_i \in \{0, 1\}$, and finally have a response $Y_i \in \bbR$. The treatment may be confounded with the unobserved noise that also affects the response, as seen in the diagram in Figure \ref{fig:dag}. 

The only aspect of this pipeline that the experimenter can vary is the nudge assignment probability $e_Z(X) = \bbP(Z = 1 \st X)$. The true propensity score $e_W(X) = \bbP(W = 1 \st X)$ is not chosen by the experimenter due to noncompliance, but it is affected by $e_Z(X)$, since
\begin{align}\label{eq:ez2ew}
e_W(X) &= \bbP(W = 1 \st X) \notag\\ 
&= \bbP(\text{always-taker} \st X) + \bbP(\text{complier, } Z = 1 \st X)\notag \\ 
&= p_{\at}(X) + p_{\co}(X)e_Z(X). 
\end{align} 

\section{Methodology} \label{sec:methodology}
In this section, we define our model, which incorporates a linear treatment effect on top of a non-parametric baseline. This ensures that the LATE is linear in the parameter vector, allowing us to use methods from optimal design of experiments. We adapt a modeling approach from \cite{robinson} to our IV setting to estimate the LATE and formulate a constrained optimization problem to select the nudge propensities. We give conditions under which the variance of the LATE using a pilot study converges to the true variance, justifying optimal design with pilot estimates.  We present an algorithm for this procedure and note some properties of the resulting optimal designs.

\subsection{Partially linear model} \label{ssec:plm}
To arrive at a criterion for optimality of $e_Z(X)$, we will use a more specific model than Equation \eqref{generalequation}. We choose a partially linear model: 
\begin{align}\label{partlin} Y = f(X) + W X^{\top}\gamma + \varepsilon. \end{align}
Here, $f$ is a non-parametric baseline model, $\gamma \in \bbR^d$ is a vector of coefficients for the treatment term, and as in Equation \eqref{generalequation}, we allow for $\bbE[\varepsilon \st X, W] \neq 0$ due to the possibility of confoundedness. 

This is a well-studied model in the causal inference literature; see \cite{robinson} and \cite{wagernotes}. A variant of this model has also been studied in the instrumental variables setting (e.g., \cite{florens, chen}) but not, to our knowledge, from a design perspective. For our purposes, it has the desirable property of allowing for flexibility in the baseline model while still preserving some parametric structure in the treatment effect that lends itself to a design criterion. 

Conveniently, the conditional LATE in the partially linear model \eqref{partlin} is $X^{\top}\gamma$, whereas this is not the case for the conditional ATE due to confoundedness. To see this, observe that
\begin{align*} 
\tau_{\late} &= \bbE_{X \st \co}\left[\frac{\bbE[Y \st X, Z = 1] - \bbE[Y \st X, Z = 0]}{p_{\co}(X)} \right] \\ 
&= \bbE_{X \st \co}\left[\frac{(\bbE[W \st X, Z = 1] - \bbE[W \st X, Z = 0])X^{\top}\gamma}{p_{\co}(X)}\right] \\ 
&\quad + \bbE_{X \st \co}\left[\frac{\bbE[\varepsilon \st X, Z = 1] - \bbE[\varepsilon \st X, Z = 0]}{p_{\co}(X)} \right] \\ 
&= \bbE_{X \st \co}\left[\frac{p_{\co}(X)X^{\top}\gamma}{p_{\co}(X)} \right] \\ 
&= \bbE_{X \st \co}[X^{\top}\gamma] \\ 
&= \bbE[X \st \co]^{\top} \gamma.
\end{align*} 
The crucial property is that $\bbE[\varepsilon \st X, Z] = \bbE[\varepsilon \st X]$ because of the conditional exogeneity of the instrument. If we have an estimator $\bar{X}_{\co}$ for $\bbE[X \st \co]$ and an estimator $\hat{\gamma}$ for $\gamma$, this immediately suggests $\bar{X}_{\co}^{\top} \hat{\gamma}$ as an estimator for $\tau_{\late}$. 

\subsection{Estimating the LATE} \label{ssec:LATEest}
We now outline a strategy to estimate $\gamma$ in the partially linear model \eqref{partlin}. Suppose that we have chosen a particular nudge propensity $e_Z(X)$, which we discuss how to do in Section \ref{ssec:design}. After the experiment has concluded, we have access to the full dataset $\{X_i, Z_i, W_i, Y_i\}_{i \leq n}$. 

To estimate $\gamma$ from this dataset, we adapt a method from \cite{robinson} to our IV setting. When there is no confounding, \cite{robinson} showed that we can regress $Y - \bbE[Y \st X]$ on $(W - e_W(X))X$ to consistently estimate the vector $\gamma$. Our setting is more complicated because the lack of compliance induces nonzero conditional means for the errors in such a regression. We address that by subtracting a more complicated expression from $Y$, denoted below by $m^*$. The resulting error term then has mean zero given ($X$, $Z$, $W$), and so the desired $\gamma$ can be estimated via ordinary least squares (OLS).

Recall from equation~\eqref{eq:ez2ew} that our choice of a nudge propensity function $e_Z$ induces the treatment propensity $e_W(X) = p_{\at}(X) + p_{\co}(X) e_Z(X)$. Define the conditional mean functions
\begin{align*} 
m(X, Z, W) &= \bbE[Y \st X, Z, W] = f(X) + WX^{\top}\gamma + \bbE[\varepsilon \st X, Z, W],\quad\text{and} \\ 
m(X, Z) &= \bbE[Y \st X, Z] = f(X) + (p_{\at}(X) + p_{\co}(X)Z)X^{\top}\gamma + \bbE[\varepsilon \st X]. 
\end{align*} 
To avoid ambiguity, the arguments to $m$ will always be explicitly given.  
In addition, let 
\begin{align} \begin{split} \label{eq: m*} m^*(X, Z, W) &= m(X, Z, W) + \frac{e_W(X) - W}{p_{\co}(X)}(m(X, Z = 1) - m(X, Z = 0)) \\ 
&= f(X) + e_W(X) X^{\top}\gamma + \bbE[\varepsilon \st X, Z, W]. \end{split} \end{align} 
Then 
\begin{align} \label{OLSeqn} Y - m^*(X, Z, W) = (W - e_W(X))X^{\top}\gamma + \varepsilon - \bbE[\varepsilon \st X, Z, W]. \end{align}
If we knew the oracle $m^*(X, Z, W)$ and $e_W(X)$, then \eqref{OLSeqn} implies that under mild conditions we could recover an unbiased and consistent estimate of $\gamma$ via OLS:
\begin{align} \label{eq:gamma_oracle}
\hat{\gamma}_{\oracle} = \bigl(X^{\top}D(X,W)^2 X\bigr)^{-1} X^{\top}D(X,W)\bigl(Y - m^*(X, Z, W)\bigr)\end{align}
where $D(X,W)=\diag(W-e_W(X))$.

We do not actually know the quantities $D$ and $m^*$ in the right side of~\eqref{eq:gamma_oracle}. Instead, we can estimate the compliance probabilities $\{p_{\co}(X), p_{\at}(X), p_{\nt}(X)\}$ and the conditional mean functions $m(X, Z, W)$ and $m(X, Z)$ from the full dataset and plug them into \eqref{eq:gamma_oracle} to obtain
\begin{align} \label{eq:gamma_approx}
\hat{\gamma}_{\plugin} = \bigl(X^{\top}\hat D(X,W)^2X\bigr)^{-1} X^{\top}\hat D(X,W)\bigl(Y - \hat{m}^*(X, Z, W)\bigr).\end{align}
Here, $\hat D(X,W)=\diag(W-\hat e_W(X))$ for $\hat{e}_W(X) = \hat{p}_{\at}(X) + \hat{p}_{\co}(X) e_Z(X)$ and
\begin{equation} \label{eq:m*hat} 
\hat{m}^*(X, Z, W) = \hat{m}(X, Z, W) + \frac{\hat{e}_W(X) - W}{\hat{p}_{\co}(X)}\left(\hat{m}(X, Z = 1) - \hat{m}(X, Z = 0)\right).
\end{equation} 

At first glance, it may seem difficult to estimate the compliance probabilities because there are some subjects whose compliance status is unknown; a unit with $(Z, W) = (1, 1)$, for instance, could be either a complier or an always-taker. However, the assumption of no defiers guarantees that, for example, a unit with $(Z, W) = (0, 1)$ is an always-taker. We can then proceed by fitting a pair of logistic regressions to estimate $\bbP(W = 1 \st Z, X)$ for $Z \in \{0, 1\}$. Note that
\begin{align*} \bbP(W = 1 \st Z = 0, X) &=  p_{AT}(X), \\ 
\bbP(W = 1 \st Z = 1, X) &=  p_{AT}(X) + p_C(X). \end{align*} 
We thus use the final estimates
\begin{align*} 
\hat{p}_{\at}(X) &= \hat{\bbP}(W = 1 \st Z = 0, X), \\ 
\hat{p}_{\nt}(X) &= \hat{\bbP}(W = 0 \st Z = 1, X),\quad \text{and} \\ 
\hat{p}_{\co}(X) &= 1 - \hat{p}_{\at}(X) - \hat{p}_{\nt}(X).
\end{align*} 

The next theorem summarizes conditions under which the difference between $\hat{\gamma}_{\plugin}$ and 
$\hat{\gamma}_{\oracle}$ converges in probability to zero for a particular choice $e_Z(X)$ of nudge propensity.

\begin{theorem}\label{thm:gamma-consistency}
Suppose that the following conditions all hold: 
\begin{enumerate} 
\item The propensity scores are estimated sufficiently well:
$$\Vert\hat{e}_W(X) - e_W(X)\Vert_{\infty} = \underset{1 \leq i \leq n}{\sup} \text{ } |\hat{e}_W(X_i) - e_W(X_i)| = o_p(1)$$
as $n\to\infty$.
\item The $m^*(X, Z, W)$ functionals are estimated sufficiently well: 
$$\frac{1}{n} \Vert\hat{m}^* - m^*\Vert^2_{2} = \frac{1}{n} \sum_{i = 1}^{n} (\hat{m}^*(X_i, Z_i, W_i) - m^*(X_i, Z_i, W_i))^2 = o_p(1),$$
as $n\to\infty$.
\item There exist $c, C > 0$ and $n_*\ge d$ with $\Vert X_i\Vert^2_2 \leq C$ for all $1 \leq i \leq n$ and $X^{\top}X/n \succeq cI_{d}$ for all $n\ge n_*$, where $\succeq$ refers to the standard partial ordering on positive semi-definite matrices.
\item There exists $\eta > 0$ such that $e_W(X) = p_{\at}(X) + p_{\co}(X) e_Z(X) \in [\eta, 1 - \eta]$ for all $X \in \mathcal{X}$. 
\end{enumerate} 
For given $(X, Z, W, Y)$ data, let $\hat{\gamma}_{\oracle}$ be the OLS estimate of $\gamma$ fit using oracle parameters as in \eqref{eq:gamma_oracle}, and let $\hat{\gamma}_{\plugin}$ be the OLS estimate fit using approximations as in \eqref{eq:gamma_approx}. Then 
\begin{align*} 
\quad \Vert\hat{\gamma}_{\oracle} - \hat{\gamma}_{\plugin}\Vert_2 \overset{p}{\to} 0 \end{align*}
as $n\to\infty$. 
\end{theorem}
The proof is in Appendix \ref{ssec:proofthmgammaconsistency}. The first two conditions in Theorem \ref{thm:gamma-consistency} concern the accuracy of our plug-in approximations. The third condition is a standard regression design constraint, and the fourth ensures overlap of the treatment and control distributions.
Condition (1) is implied if the analogous conditions hold for each of $\{\hat{p}_{\co}(X), \hat{p}_{\at}(X), \hat{p}_{\nt}(X)\}$. Condition (4) holds, for instance, if Assumption \ref{as:somecompliers} is strengthened to a uniform lower bound on $p_C(X)$ and $e_Z(X)$ is uniformly bounded away from zero and one. 

In Appendix \ref{sec:cross-fitting}, we also outine a variant of this procedure that uses cross-fitting to obtain a $\sqrt{n}$ convergence rate for $\hat{\gamma}_{\plugin}$. This method requires stronger assumptions, is more complicated, and is standard in the cross-fitting literature in causal inference, so we defer it to the appendix for simplicity. 

To estimate $\bbE[X \st \co]$, we note that 
$$\frac{1}{n} \sum_{i = 1}^{n} \frac{X_i p_{\co}(X_i)}{p_{\co}}$$
would be a consistent estimator, where $p_{\co}$ is the marginal probability of being a complier. Of course, the true $p_{\co}(X_i)$ and $p_{\co}$ are not known, so we define the estimator
\begin{equation} \label{eq:Xbarc} \bar{X}_{\co} = \frac{1}{n} \sum_{i = 1}^{n} \frac{X_i \hat{p}_{\co}(X_i)}{\hat{p}_{\co}}. \end{equation}
Here, $\hat{p}_{\co}$ is an estimate of the marginal compliance probability (e.g., by averaging all $\hat{p}_{\co}(X_i)$ terms). This too is consistent under mild assumptions. For instance, if 
$$||\hat{p}_C(X) - \hat{p}_{\co}(X)||_{\infty} = \underset{1 \leq i \leq n}{\sup} \text{ } |\hat{e}_W(X_i) - e_W(X_i)| = o_p(1)$$ 
and the coordinates of $X_i$ are all finite variance, then for $1 \leq j \leq d$:
$$\frac{1}{n} \left|\sum_{i = 1}^{n} \frac{X_{ij}\hat{p}_{\co}(X_i)}{\hat{p}_{\co}} - \frac{X_{ij}p_{\co}(X_i))}{p_{\co}}\right| \leq \frac{1}{n} \underbrace{\left\|\frac{\hat{p}_{\co}(X)}{\hat{p}_{\co}} - \frac{p_{\co}(X)}{p_{\co}}\right\|_{\infty}}_{o_p(1)} 
\cdot
\underbrace{\sum_{i = 1}^{n} |X_{ij}|}_{O_p(n)} = o_p(1).$$
Hence, $\overline{X}_{\co}$ is then consistent as well. 

The estimates $\bar{X}_{\co}$, $\hat{m}^*$, and  $\hat e_W$ can be computed on the experimental data, and we assume above that they are consistent.  However, we may want more accuracy by optimizing the nudge propensities $e_Z$.
In the next section, we explore how to choose a particular $e_Z(X)$ to minimize the variance of this procedure.

\subsection{Designing for nudge assignment} \label{ssec:design}
We now consider the question of how to choose the nudge propensity $e_Z(X)$ for the procedure outlined in the previous section. To obtain a design criterion for choosing $e_Z(X)$, we assume access to a pilot study of size $n_0$ that has been run prior to the main study under some nudge propensities that we denote by $e_Z^0(X)$. In practice, $e_Z^0(X)$ would likely have to satisfy any design constraints imposed on the optimal solution as well.

For this pilot study, we assume that nudges have been assigned and treatment statuses have been recorded. The optimal design will depend only on the pilot estimates of the compliance probabilities, and not on $\hat m^*(X,Z,W)$. As a result, we do not need the response values $Y_i$ from the pilot study at the time we design the main study, which is valuable in settings where $Y_i$ are observed after a long wait. 

We thus assume that we have a pilot dataset $\{(X_i, Z_i, W_i)\}_{i \leq n_0}$ available before the main study. We assume that the covariates and compliance classes come IID from the same joint distribution as in the main study. For notational clarity, we keep this dataset separate from the main study $\{(X_i, Z_i, W_i, Y_i)\}_{i \leq n}$. In the analysis below, the design matrix $X$ and the vectors $\bar{X}_{\co}$, $Z$, and $W$ all refer only to the data in the main study, and not to the pilot, unless otherwise stated. 

To obtain a design criterion, assume for now that $\Var(\varepsilon - \bbE[\varepsilon \st X, Z, W] \st X, Z, W) = \sigma^2$ is constant in $X$, $Z$, and $W$. We discuss generalizations to heteroscedasticity at the end of this section and in further detail in Appendix \ref{sec:WLS}. Using the pilot study, we construct estimates $\{\hat{p}_{\co}(X_i), \hat{p}_{\at}(X_i), \hat{p}_{\nt}(X_i)\}_{i \leq n}$ for the compliance probabilities in the main study, as described in Section \ref{ssec:LATEest}. 

After the pilot study, a natural variance criterion to minimize is 
\begin{align} \label{c-opt}
\Var(\bar{X}_{\co}^{\top} \hat{\gamma} \st X) = \bar{X}_{\co}^{\top} \Var(\hat{\gamma} \st X) \bar{X}_{\co}.
\end{align}
We condition on $X$ in \eqref{c-opt} because the $X$ data is not chosen in our setting but rather given. A design that minimizes \eqref{c-opt} satisfies a particular case of $C$-optimality \citep{atkinson2007optimum}, which minimizes the quadratic form of a vector $c$ and a covariance matrix. Here, $c = \bar{X}_{\co}$. Our machinery also works for other design criteria based on $\Var(\hat{\gamma} \st X)$, such as A-optimality or D-optimality. 

From the formula for $\hat{\gamma}_{\text{oracle}}$ in \eqref{eq:gamma_oracle}, it follows that
$$\Var(\hat{\gamma}_{\oracle} \st X, Z, W) = \sigma^2 (X^{\top}\diag(W - e_W(X))^2X)^{-1}.$$ 
We thus define the oracle variance, which is conditional on only $X$, as
\begin{align*} 
V_{\oracle}(e_Z) &= \Var(\sqrt{n}\, \bar{X}_{\co}^{\top}\hat{\gamma}_{\oracle} \st X) \\ 
&= n \bar{X}_{\co}^{\top}\bbE\bigl[\Var(\hat{\gamma}_{\oracle} \st X, Z, W) \st X\bigr]\bar{X}_{\co} + n\bar{X}_{\co}^{\top}\Var(\bbE[\hat{\gamma}_{\oracle} \st X, Z, W] \st X)\bar{X}_{\co} \\ 
&= n\sigma^2 \bar{X}_{\co}^{\top}\bbE\bigl[(X^{\top}\diag(W - e_W(X))^2X)^{-1} \st X\bigr]\bar{X}_{\co}. \end{align*}
Here, we use that $\hat{\gamma}_{\oracle}$ is unbiased for $\gamma$ to eliminate the second term in the law of total variance. We also scale by $\sqrt{n}$ since the unscaled variance in OLS converges to zero. Because $\sigma^2$ is merely a multiplicative constant,  we take $\sigma^2 = 1$ henceforth without loss of generality. 

With the above assumptions, a technical difficulty remains because there is a matrix inverse inside of an expectation. To handle this, we make our first of two approximations by moving the inverse out of the expectation:
\begin{align*} 
V_{\oracle}(e_Z)
&= n\bar{X}_{\co}^{\top}\bbE[(X^{\top}\diag(W - e_W(X))^2X)^{-1} \st X]\bar{X}_{\co} \\ 
&\approx n\bar{X}_{\co}^{\top}\bbE[X^{\top}\diag(W - e_W(X))^2X \st X]^{-1}\bar{X}_{\co} \\ 
&= n\bar{X}_{\co}^{\top}(X^{\top}\diag(e_W(X)(1 - e_W(X))X)^{-1}\bar{X}_{\co} \\
&= \widetilde V_{\oracle}(e_Z). 
\end{align*} 
We use here that, with $X$ fixed, $\bbE_W[(W - e_W(X))^2 \st X]$ is the variance of a Bern$(e_W(X))$ random variable, so it equals $e_W(X)(1 - e_W(X))$. 

Finally, we make one more approximation to arrive at something computable by replacing the oracle $e_W(X)$ with $\hat{e}_W(X)$ estimated from the pilot study: 
\begin{align*} 
\widetilde V_{\oracle}(e_Z) &= n\bar{X}_{\co}^{\top}\bigl(X^{\top}\diag(e_W(X)(1 - e_W(X))X\bigr)^{-1}\bar{X}_{\co} \\ 
&\approx n\bar{X}_{\co}^{\top}\bigl(X^{\top}\diag(\hat{e}_W(X)(1 - \hat{e}_W(X))X\bigr)^{-1}\bar{X}_{\co} \\ 
&= \widetilde{V}_{\plugin}(e_Z).
\end{align*} 
For given $e_Z(X)$, recall that $\hat{e}_W(X) = \hat{p}_{\at}(X) + \hat{p}_{\co}(X) e_Z(X)$. The function $\widetilde{V}_{\plugin}(e_Z)$
is computable given just the new $X$ data and the compliance probability estimates from the pilot study. Moreover, for a given sample, it is convex in the quantities $\{e_Z(X_1),\dots, e_Z(X_n)\}$, as we verify in Appendix \ref{ssec:proofconvexity}. Hence, it can be efficiently minimized subject to any convex constraints. That is, with $\mathcal{C}_n \subseteq [0, 1]^n$ being the feasible region for $e_Z$ imposed by our convex constraints, we can solve for 
\begin{align}\label{eq:convprogram} \hat{e}_Z^* = \underset{e_Z \in \mathcal{C}_n}{\text{argmin }} \widetilde{V}_{\plugin}(e_Z).
\end{align}
The two approximations that yield a tractable convex optimization problem are visualized in Figure \ref{fig:var-approx}, and we summarize our method in Algorithm \ref{alg:mainalgo}. We emphasize that, once all response data has been collected, the pilot and main study can be combined into a full dataset and used for the estimation procedure described in Section \ref{ssec:LATEest}.

We could analogously define $e_Z^*$ as the design that minimizes the true variance $V_{\oracle}(e_Z)$ over $\mathcal{C}_n$, though there is no way to solve for this. In light of our approximations, a desirable property of our method is that, asymptotically, using $\hat{e}_Z^*$ recovers the optimal true variance $V_{\oracle}(e_Z^*)$ and the oracle estimate $\hat{\gamma}_{\oracle}$. The following theorem confirms that this is the case under certain assumptions. 

\begin{theorem}\label{thm:Vopt}
Suppose that the following conditions all hold: 
\begin{enumerate} 
\item The pilot study has size $n_0 = \Omega(n)$ that is not slower than linear in $n$.
\item The propensity scores are estimated sufficiently well in the pilot study:
$$\Vert\hat{e}_W(X) - e_W(X)\Vert_{\infty} = \underset{1 \leq i \leq n}{\sup} \text{ } |\hat{e}_W(X_i) - e_W(X_i)| = o_p(1)$$
as $n_0 \to\infty$.
\item There exist $c, C > 0$ and $n_*\ge d$ with $\Vert X_i\Vert^2_2 \leq C$ for all $1 \leq i \leq n$ and $X^{\top}X/n \succeq cI_{d}$ for all $n\ge n_*$, where $\succeq$ refers to the standard partial ordering on positive semi-definite matrices.
\item There exists $\eta > 0$ such that $e_W(X) = p_{\at}(X) + p_{\co}(X) e_Z(X) \in [\eta, 1 - \eta]$ for all $X \in \mathcal{X}$ and $e_Z \in \mathcal{C}_n$. 
\end{enumerate} 
Let $e_Z^*$ and $\hat{e}_Z^*$ minimize the true variance $V_{\oracle}(e_Z)$ and estimated variance $\widetilde{V}_{\plugin}(e_Z)$ respectively over $\mathcal{C}_n$. Then 
\begin{align*} |V_{\oracle}(\hat{e}_Z^*) - V_{\oracle}(e_Z^*)| &\overset{p}{\to} 0 \end{align*}
as $n\to\infty$. 
\end{theorem}

\begin{algorithm}[t!]
    \caption{Procedure to choose nudge propensities}\label{alg:mainalgo}
    \begin{algorithmic}[1]
	\STATEx { \textbf{Input:} pilot study $\mathcal{P} = \{X_i, Z_i, W_i\}_{i \leq n_0}$ obtained via $e_Z^0$, new data $X \in \bbR^{n \times d}$, feasible region $\mathcal{C}_n \subseteq [0, 1]^n$ for $e_Z$ defined via convex constraints.}
        \STATE{Using $\mathcal{P}$, construct estimates $\{\hat{p}_{\co}(X_i), \hat{p}_{\at}(X_i), \hat{p}_{\nt}(X_i)\}_{i \leq n}$ of the compliance probabilities.}  
        \STATE{With $\hat{D}(e_Z) = \hat{p}_{\at}(X) + \hat{p}_{\co}(X) e_Z(X)$ and $\bar{X}_{\co}$ as in \ref{eq:Xbarc}, define}
        \STATEx{\qquad $\widetilde{V}_{\plugin}(e_Z) = n \bar{X}_{\co}^{\top} (X^{\top}\hat{D}(e_Z)(1 - \hat{D}(e_Z))X)^{-1}\bar{X}_{\co}$.}
        \STATE{Solve the convex program \eqref{eq:convprogram} for $\hat{e}_Z^* \in \bbR^n$.}
        \STATE{Assign nudges with probabilities $\hat{e}_Z^*$
        to observe $\{Z_i, W_i\}_{i \leq n}$, and collect all responses $Y_i \in \bbR^{n_0 + n}$.}
        \STATE{Merge the pilot and main studies into a single dataset $\{X_i, Z_i, W_i, Y_i\}_{i \leq n_0 + n}$.}
        \STATE{Using the full dataset, construct an estimate $\hat{m}^*(X, Z, W)$ of the function $m^*(X, Z, W)$ defined in \eqref{eq: m*}, and construct improved estimates of the compliance probabilities.}
        \STATE{With $\hat{D}(X, W) = \diag(W - (\hat{p}_{\at}(X) + \hat{p}_{\co}(X) e_Z^*(X))) \in \bbR^{d \times d}$ computed from the full dataset, construct}
        \STATEx {\qquad $\hat{\gamma}_{\plugin} = \bigl(X^{\top}\hat{D}(X, W)^2X\bigr)^{-1}X^{\top}\hat{D}(X, W)\bigl(Y - \hat{m}^*(X, Z, W)\bigr) \in \bbR^d$.}
        \STATEx{}
        \STATEx{\textbf{Output:} $\hat{e}_Z^*$, $\hat{\gamma}_{\plugin}$.}
	\end{algorithmic}
\end{algorithm}

\begin{figure}[t!] 
\centering
\includegraphics[width=8cm]{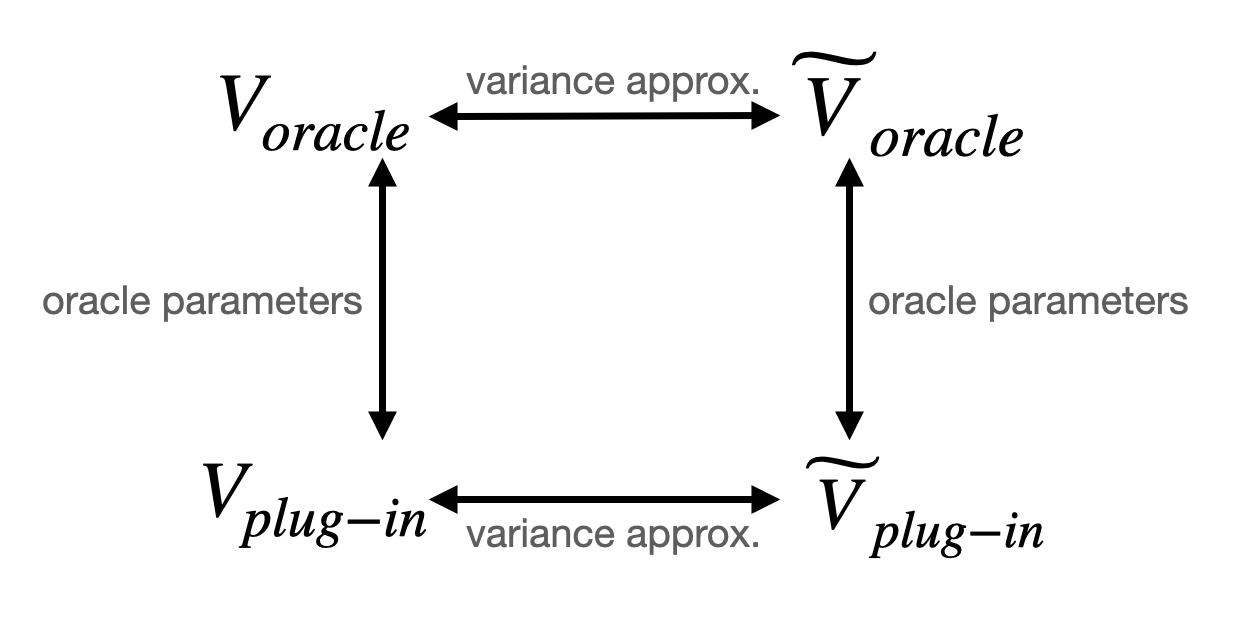}
\caption{Visualization of the approximations that move from $V_{\oracle}$ to $\widetilde{V}_{\plugin}$. The left-right transition involves moving a matrix inverse out of an expectation to approximate the variance, and the top-bottom transition involves replacing oracle compliance probabilities with their estimates from the pilot study.}
\label{fig:var-approx}
\end{figure}

The proof is in Appendix \ref{ssec:proofthmvopt}. Theorem \ref{thm:Vopt} confirms that the approximations made to arrive at a solvable design problem do no harm asymptotically. The main additional assumption relative to Theorem \ref{thm:gamma-consistency} is that the compliance probabilities are estimated sufficiently well not just in the main study population, but in the pilot study. To make this more believable, we also add the assumption that the ratio of the pilot study size and the main study size is bounded below asymptotically. 

It is possible to generalize this design procedure to heteroscedastic errors that depend on $X$ and $W$ by estimating the conditional variances $\sigma^2(X, W) = \Var(\varepsilon \st X, W)$ in the pilot study. These can then be incorporated to change the second-stage regression to a weighted least squares procedure. However, this requires estimating the function $m^*(X, Z, W)$ in the pilot study as well. For ease of exposition, we defer full details of this method to Appendix \ref{sec:WLS}. Of note, it yields a convex variance criterion, as in the homoscedastic case, provided that the variance ratio $\sigma^2(X, W = 1)/\sigma^2(X, W = 0)$ is always bounded in $[1/2, 2]$ (see Theorem \ref{thm:hetero_ratio} in Appendix \ref{sec:WLS}). 

\subsection{Constraints on nudge assignment} \label{ssec:constraints}
A natural question to ask is what the optimal solution $\hat{e}_Z^* \in [0, 1]^n$ looks like under various constraint settings. The results in this section investigate this question in the homoscedastic setting. This question has been considered for fully linear models without noncompliance by \cite{liowenTBD} and \cite{morrisonowen}. They prove that, when a monotonicity constraint is imposed that forces $e_{W}(X_1) \leq e_{W}(X_2) \leq \cdots \leq e_{W}(X_n)$, optimal designs consist of a few discrete levels of constant probability. \cite{liowenTBD} also find that optimal solutions in a simple one-dimensional linear model have a few levels under only budget and gain constraints. We discuss these particular constraints in greater detail in Section \ref{sec:triage}.

As in those papers, we find that the RCT that has $e_W(X) = 1/2$ everywhere is optimal, if attainable, in the unconstrained setting. Conceptually, this is because the $e_W(X)(1 - e_W(X))$ entries in the diagonal matrix in $V_{\plugin}(e_Z)$ are all minimized at $e_W(X) = 1/2$. Under noncompliance, this equates to choosing
\begin{align} \label{eq:RCTtwostage} \hat{p}_{\at}(X) + \hat{p}_{\co}(X)\hat{e}^*_Z(X) = \frac{1}{2} \implies \hat{e}_Z^*(X) = \frac{1 - 2\hat{p}_{\at}(X)}{2\hat{p}_{\co}(X)}. \end{align}
In short, although we cannot choose $e_W(X)$ directly, to minimize variance we would choose $e_Z(X)$ so that we anticipate getting an RCT in the second stage based on our knowledge of the compliance probabilities. 

Equation \eqref{eq:RCTtwostage} implements this strategy by assigning lower nudge propensity where the rate of being an always-taker is high and assigning higher nudge propensity where the rate of being a never-taker is high. However, even in the unconstrained setting, this choice may be impossible if the (estimated) rate of noncompliance is too high, i.e., if 
$$\frac{1 - 2\hat{p}_{\at}(X_i)}{2\hat{p}_{\co}(X_i)} \not \in [0, 1]$$ 
for some $X_i$. In this case, the solution is to threshold to whichever of $0$ or $1$ is closer to this ratio. This choice only accounts for statistical considerations and may be infeasible on economic or ethical grounds.

Interestingly, a similar result holds under a budget constraint, which forces the average estimated treatment probability to be some amount. This can be stated mathematically as
$$\bar{\hat{e}}_W = \frac{1}{n} \sum_{i = 1}^{n} \hat{p}_{\at}(X_i) + \hat{p}_{\co}(X_i) e_Z(X_i) = \mu \in [0, 1].$$ 
This constraint ensures that a fixed fraction of subjects get the nudge in expectation and is useful if there is a limited supply of treatment available. In the case where there is no noncompliance (or, more generally, when the compliance probability $\hat{p}_{\co}(X_i)$ is constant), we find that the optimal solution is the RCT with $\hat{e}_W(X) = \mu$. This result requires that $X$ have an intercept in its span. Interestingly, this differs from the findings in \cite{liowenTBD} and \cite{morrisonowen}, whose objectives and models do not yield constant probabilities once a budget constraint has been added. This difference arises from special properties of the partially linear model and the particular objective function $\widetilde{V}_{\plugin}(e_Z)$ that it induces. We summarize all of this in the theorem below.
\begin{theorem}\label{thm:constraintresults}
Consider the optimization problem 
$$\hat{e}_Z^* = \underset{e_Z \in \mathcal{C}_n}{\arg\min} \text{ } \widetilde{V}_{\plugin}(e_Z),$$
where $\mathcal{C}_n \subseteq [0, 1]^n$ is a region defined by convex constraints and $\widetilde{V}_{\plugin}$ is as in Section \ref{ssec:design}. \\ 
(a) If $\mathcal{C}_n = [0, 1]^n$ is unconstrained, then an optimal solution is to take 
$$\hat{e}_Z^*(X_i) = \begin{cases}\frac{1 - 2\hat{p}_{\at}(X_i)}{2\hat{p}_{\co}(X_i)}, & \hat{p}_{\at}(X_i) \leq \frac{1}{2} \text{ and } \hat{p}_{\at}(X_i) + \hat{p}_{\co}(X_i) \geq \frac{1}{2} \\ 
0, & \hat{p}_{\at}(X_i) > \frac{1}{2} \\ 
1, & \hat{p}_{\at}(X_i) + \hat{p}_{\co}(X_i) < \frac{1}{2}.
\end{cases} $$
(b) If (i) $\mathcal{C}_n = \{e_Z \in [0, 1]^n \st \bar{e}_W = \mu\}$ for some $\mu \in [0, 1]$ and (ii) the compliance probability $\hat{p}_{\co}(X)$ is constant and (iii) $1_n \in \mathrm{span}(X)$, then an optimal solution is to take
$$\hat{e}_Z^*(X_i) = \frac{\mu - \bar{\hat{p}}_{\at}}{\hat{p}_{\co}},$$
where $\bar{\hat{p}}_{\at}$ is the sample average of the always-taker probabilities, provided that this fraction is in $[0, 1]$. 
\end{theorem}
Theorem \ref{thm:constraintresults} provides intuition about the nature of optimal solutions under minimal constraints. In the unconstrained case, an increase in $\hat{p}_{\at}(X_i)$ coincides with a decrease in $\hat{e}_Z^*(X_i)$ to try to balance the number of treated and control units, for instance.

\section{Triage example}\label{sec:triage} 
In this section, we outline an example application of this procedure that is motivated by the problem of patient triage in emergency departments. When ED patients are admitted to the hospital, they are assigned a level of care based on their perceived need and bed availability, with the highest level of care being an intensive care unit (ICU). Because ICUs require greater resources and are constrained in their bed and nursing availability, a natural question of interest is the causal effect of being sent to the ICU on a patient's health outcomes. This question could impact decisions about both the size of the ICU \citep{phua2020lessmore} and future triaging, which may be able to exploit the heterogeneous effects of treatment \citep{chang2017icu}.

Because full randomization would be highly unethical, we consider a semi-synthetic example based on a study design in which the experimenter assigns a given patient to either the ICU ($Z = 1$) or less intensive care ($Z = 0$) and their doctor can override if deemed necessary. This example thus falls naturally into our framework because the experimenter's proposed triage is a nudge $Z$ that can depend on the patient's observed covariates $X$ and that need not agree with the final triage assignment $W$. Moreover, it is natural to expect that patients whose encouragement is overridden by the physician likely differ in some unmeasured way, introducing confoundedness.
 
We use a de-identifed electronic health records dataset that consists of $7281$ adult inpatient admissions at an academic hospital and a trauma center between 2015 and 2019. This dataset includes each patient's length of stay in the ED. We take this to be our response variable of interest, with a lower length of stay indicating a preferred outcome. It also includes several covariates about the patient, including their demographic information and various vital signs. 

To these variables, we add the risk score developed by \cite{nguyen2021}. This is a machine learning-based feature that was constructed via an ensemble of gradient-boosted decision trees, which input a patient's covariates and output their predicted probability of being sent to the ICU within the next three hours.  This variable is useful for such an analysis because it provides an estimate of whether the patient will eventually need to go to the ICU. In our motivating setting, this would be a natural variable to use to constrain the experimental design, with high-risk patients having a higher probability of encouragement than low-risk patients. For this reason, we return to it when imposing constraints later in our analysis. 

Because our dataset is observational and does not come from the proposed study design, we do not have access to genuine noncompliance data, and our example is necessarily semi-synthetic. Nonetheless, we use the observed data to fit sensible models for $Y$ given $X$ and $W$, which we take as a starting point for our analysis. Specifically, we proceed as follows: 
\begin{enumerate} 
\item Use the non-ICU data ($W = 0$) to fit a flexible baseline model $Y \sim f(X)$ via a random forest. Here, $Y \in \bbR^n$ is the length of stay (in hours) for $n$ patients and $X \in \bbR^{n \times 7}$ is a dataset containing a subset of covariates deemed particularly important. These are: hours in the ED prior to triage, weight, age, systolic blood pressure (SBP), pulse, risk score, and an intercept. One could also use a much richer class of covariates in this step and restrict to fewer for the linear model to follow, but for simplicity we opted to keep the dimensions the same throughout.    
\item Use the ICU data ($W = 1$) to fit a regression model $Y - f(X) \sim WX\gamma$, where $X \in \bbR^{n \times 7}$ is as above.  
\item Choose functions $\{p_{\co}(X), p_{\at}(X), p_{\nt}(X)\}$ for the noncompliance probabilities and assign compliance classes to patients via these probabilities. These are entirely synthetic, since there were no nudges and hence no notion of noncompliance in the data. For simplicity and interpretability, we choose these to depend only on the risk score of \cite{nguyen2021}. We base our choice on the heuristic that it is natural to expect higher rates of noncompliance at extremes of this risk score. Full formulas are available in Appendix \ref{sec:expdetails}.
\item Choose a distribution for the error term $\varepsilon$ that depends on compliance class. We take $\varepsilon$ to include additive terms for each noncompliance class, i.e., 
$$\varepsilon_i \overset{\mathrm{ind}}{\sim} \indic\{\text{AT}_i\}\gamma_{\at} + \indic\{\text{NT}_i\}\gamma_{\nt} + N(0, \sigma^2).$$
The particular coefficients $\gamma_{\at}$ and $\gamma_{\nt}$ are free parameters that tune how much the noncompliance affects outcomes. In addition, $\sigma^2$ is a noise parameter that can be tuned.
\end{enumerate} 
Further details about the simulation procedure are available in Appendix \ref{sec:expdetails}. 

In Figure \ref{fig:opt-sols}, we present the optimal solution $\hat{e}_Z^*$ under an array of constraint settings interpreted below. We used $n = 6000$. The risk scores in that figure have been rank transformed to have a uniform distribution. The constraints imposed, in various combinations, are:
\begin{enumerate} 
\item No constraints, i.e., $\mathcal{C}_n = [0, 1]^n$. Here, the solution is known perfectly via Theorem \ref{thm:constraintresults} and convex optimization is not needed. 
\item Monotonicity constraint that $\hat{e}_Z(X)$ is nondecreasing in the risk score $r(X)$. This serves as a fairness constraint to force the nudge propensity to be no lower for a higher risk patient than for a lower risk one.
\item Budget constraint with $\mathcal{C}_n = \{e_Z \in [0, 1]^n \st \overline{p_{\at}} + \overline{p_{\co} e_Z} = 0.3\}$. Under this constraint, the expected proportion of subjects treated is $30\%$.
\item Gain constraint that 
$$\sum_{i = 1}^{n} e_Z(X_i) r(X_i) \geq \rho \sum_{i = k}^{n} r(X_{(i)}),$$ 
where $\rho$ is a free parameter and the parenthetical subscript on the right-hand side indicates that we sort risk scores from smallest to largest. In words, the expected total risk among those patients encouraged is at least some fraction $\rho$ times what it would be under a regression discontinuity design (RDD) in which the encouragement is given deterministically to the patients with the highest risk score. The lower summand $k$ on the right-hand side is dictated by the other constraints of the problem; for example, under a budget constraint, $k$ would be such that the average encouragement probability across the whole population of patients is as stipulated. 

This constraint simulates one that might arise in a real setting, in which there is some expected benefit to be derived by giving the encouragement to particular units over others. Here, the risk score of \cite{nguyen2021} is used as a proxy for this benefit. If the $\rho$ parameter were brought closer to $1$, the result would tend to a regression discontinuity design in which precisely the top patients by risk score were given the encouragement. 
\end{enumerate} 

Figure \ref{fig:opt-sols} shows the graphs for optimal solutions in these settings. The unconstrained solution is clearly inappropriate in a real ICU, since it assigns the highest nudge propensities to cases with the lowest risk scores and vice versa. This optimizes statistical efficiency by assigning more encouragements to patients more likely to be never-takers and fewer to patients more likely to be always-takers in order to arrive at something closer to an RCT. However, it recommends the ICU more for those who need it least or not at all, and it also does not account for the limited capacity of the ICU. Adding a monotonicity constraint flattens the curve. 

\begin{figure}[t!] 
\centering
\includegraphics[width=12cm]{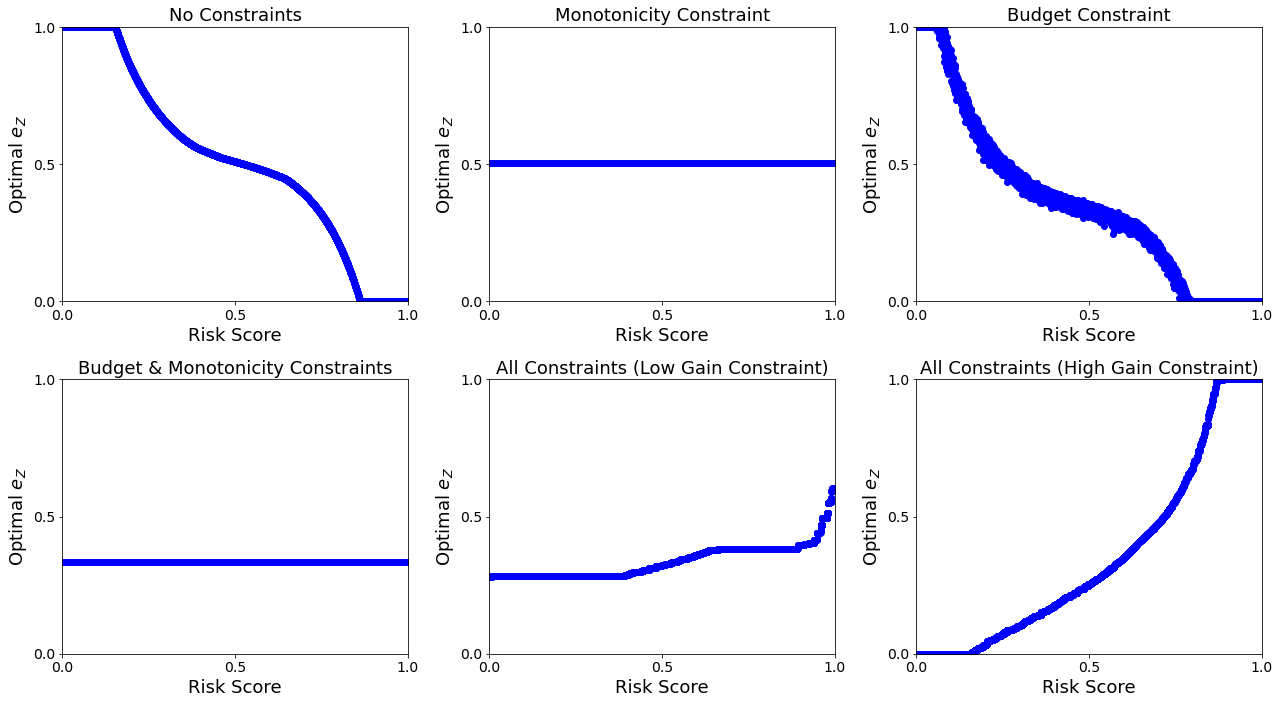}
\caption{Results of our convex optimization procedure under an array of constraint settings: (a) unconstrained, (b) monotonicity constraint, (c) budget constraint, (d) budget and monotonicity constraints (e) budget, monotonicity, and low gain constraint, and (f) budget, monotonicity, and high gain constraint. All solutions besides (a) are obtained via the CVXPY package \citep{cvxpy}.} 
\label{fig:opt-sols}
\end{figure}

The budget constraint, which allows only $40\%$ expected treatment, depresses encouragement probabilities, though the benefit remains to encourage lower risk patients more often to combat the fact that there are more never-takers in that regime. We note that a budget constraint makes SUTVA more believable in this problem; if a busier ICU provides a lower quality of care, then one patient's ICU admission status could affect the outcome of another patient. However, a budget constraint ensures that the ICU stays at a roughly fixed capacity regardless of which particular patients are treated. 

Finally, the last two panels add, respectively, a ``low" gain constraint with $\rho = 0.85$ and a ``high" gain constraint with $\rho = 0.95$. As anticipated, the solution moves closer to an RDD as $\rho$ tends to one. We also compare the values of the objective for each solution in Appendix \ref{sec:expdetails}. We find that the constrained problems have values of the objective function that are between $6\%$ and $87\%$ larger than that of the unconstrained problem. 

Separately, in Figure \ref{fig:baseline-comp} we assess the variance and mean-squared error (MSE) of our procedure to estimate $\tau_{\late}$. We do so under both a budget constraint that $\overline{\hat{e}}_W = 0.4$ and a monotonicity constraint to capture bed availability and fairness constraints that might naturally arise. As a baseline, we compare our optimal design $\hat{e}_Z^*$ to the RDD that gives encouragement deterministically to the highest risk patients, with the cutoff determined by satisfying the same $40\%$ budget constraint. This RDD is a natural comparison point that one might default to in a hospital if not using our method. At each of $n \in \{1000, 2000, 3000, 4000, 5000, 6000\}$, we compute $\hat{\tau}_{\late}$ under both our design and the baseline design. We repeat this over $250$ simulations. Finally, we compute the sample variance and MSE of each procedure and compare. 

These results are presented in Figure \ref{fig:baseline-comp}, which plots the variances and MSEs at several sample sizes. We compare the estimators from our method with no constraints and with the budget constraint at $\bar{e} = 0.4$, as well as the RDD at the corresponding budget. We see that our method consistently outperforms the baseline RDD in both variance and MSE across all values of $n$ tested. 

The results in Figure \ref{fig:baseline-comp} suggest that our method is capable of producing meaningful gains over less thoughtful designs. Moreover, when more constraints are added, it becomes less clear what a baseline design choice should even be, since identifying a feasible design becomes non-obvious.

\begin{figure}[t!]
  \centering
  \begin{subfigure}{0.48\textwidth}
    \centering
    \includegraphics[width=\textwidth]{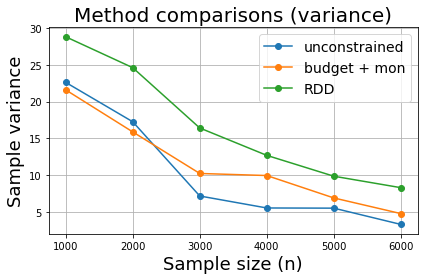}
  \end{subfigure}
  \hfill
  \begin{subfigure}{0.48\textwidth}
    \centering
    \includegraphics[width=\textwidth]{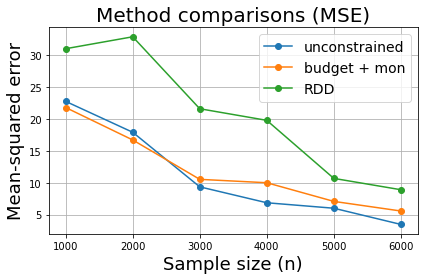}
  \end{subfigure}
  \caption{Sample variance and sample MSE of several designs. The plot shows our method with no constraints, our method with the budget constraint of $\overline{\hat{e}}_W = 0.4$ and a monotonicity constraint, and the RDD at the same budget constraint.}
  \label{fig:baseline-comp}
\end{figure}

\section{Discussion} \label{sec:discussion}
In this paper, we introduce a method to design for optimal assignment of a binary ``nudge" in a randomized encouragement design. This problem arises in a host of settings -- such as medical experiments and testing certain features on tech platforms -- where the experimenter cannot directly force a particular treatment assignment. In such cases, it is useful to be able to reason about how best to assign a nudge to improve statistical efficiency. 

Our method assumes a partially linear model, with a nonparametric baseline and a linear treatment term, which balances flexibility with optimal design compatibility. We also allow for the incorporation of convex constraints, which can capture various economic and ethical considerations that arise in real problems. In particular, a budget constraint can handle the case in which there is a finite supply of treatment available, such as limited bed capacity in an ICU. 

For an optimal design, we have used a two-stage strategy with a pilot study to estimate the compliance probabilities. One could usefully consider a multi-stage strategy where the compliance probabilities are updated on the fly.  Optimizing that approach is outside the present scope.

In Section \ref{sec:triage}, we have also presented a semi-synthetic example involving randomized encouragement to an ICU. We emphasize that this example is a proof of concept rather than an exact demonstration of what such a study would look like. In practice, the constraints that arise may be more complicated in a variety of ways, such as time-dependence. Nevertheless, we view this example as a promising demonstration of the benefits that our method can provide.

\section*{Acknowledgments} 
The authors thank John Cherian, Isaac Gibbs, Samir Khan, Harrison Li, and Anav Sood for helpful discussions and comments. T.\ M.\ was partly supported by a B.\ C.\ and E.\ J.\ Eaves Stanford Graduate Fellowship. M.\ N.\ was supported by the Stanford DARE fellowship. A.~B.~O.\ was supported by  grant DMS-2152780 from the U.S.\ National Science Foundation.

\bibliographystyle{plainnat}
\bibliography{ICU}

\appendix
\section{Proofs}\label{sec:proofs}

\subsection{Proof of LATE identifiability}\label{ssec:proofLATE} 
Here, we reprove Equation \eqref{LATEformula} for completeness. We have 
\begin{align*} \tau_{\late} &= \bbE[Y(1) - Y(0) \st \text{complier}] \\ 
&= \bbE_{X \st \co}\left[\frac{\bbE[(Y(1) - Y(0))\indic\{\text{complier}\} \st X]}{p_{\co}(X)}\right] \quad \text{(Bayes' theorem)} \\ 
&= \bbE_{X \st \co}\left[\frac{\bbE[Y(W(1)) - Y(W(0)) \st X]}{p_{\co}(X)}\right] \quad \text{(No defiers)} \\ 
&=  \bbE_{X \st \co}\left[\frac{\bbE[Y(W(1)) \st Z = 1, X] - \bbE[Y(W(0)) \st Z = 0, X]}{p_{\co}(X)}\right] \quad \text{(Conditional exogeneity)} \\ 
&= \bbE_{X \st \co}\left[\frac{\bbE[Y \st X, Z = 1] - \bbE[Y \st X, Z = 0]}{\bbE[W \st X, Z = 1] - \bbE[W \st X, Z = 0]} \right] \text{(Exclusion)}.
\end{align*}
The third line uses Assumption~\ref{as:nodefiers} that there are no defiers, so the only units for which $Y(W(1)) \neq Y(W(0))$ are compliers. Exclusion and conditional identifiability come from Definition~\ref{def:IV} for our instrumental variable. The relevance condition ensures that $p_{\co}(X) > 0$ everywhere. We used Assumption~\ref{as:somecompliers} throughout to keep 
$p_{\co}(X) > 0$. \hfill$\blacksquare$

\subsection{Proof of convexity of \texorpdfstring{$\widetilde{V}_{\plugin}(e_Z)$}{Proof of convexity of widetilde V plug-in (eZ)}} \label{ssec:proofconvexity}
Recall that 
$$\widetilde{V}_{\plugin}(e_Z) = n \bar{X}_{\co}^{\top} \left(X^{\top} \text{diag}(\hat{e}_W(X)(1 - \hat{e}_W(X))X\right)^{-1}\bar{X}_{\co}.$$
For fixed vector $x$ and matrix $X$, let 
\begin{align*} f(v) &= x^{\top}(X^{\top} \diag(v(1 - v))X)^{-1} x \\
&=: x^{\top} (X^{\top} D(v) X)^{-1} x. \end{align*} 
For our problem, $v = \hat{p}_{\at} + \hat{p}_{\co} e_Z$, and we are optimizing with respect to $e_Z$. However, this is just a linear offset, so it suffices to prove convexity of $f(v)$.

Note that $D(v)$ is concave in $v$ in the sense that 
$$D(\lambda u + (1 - \lambda) v) \succeq \lambda D(u) + (1 - \lambda) D(v),$$
which follows by concavity of the map $v \mapsto v(1 - v)$. Then under the standard partial ordering on positive semidefinite matrices, 
$$X^{\top} D(\lambda u + (1 - \lambda) v)X \succeq \lambda X^{\top} D(u) X + (1 - \lambda) X^{\top} D(v) X$$
and
$$(X^{\top} D(\lambda u + (1 - \lambda) v)X)^{-1} \preceq (\lambda X^{\top} D(u) X + (1 - \lambda) X^{\top} D(v) X)^{-1}.$$
Hence,  
\begin{align*} 
f(\lambda u + (1 - \lambda)v) &= x^{\top} (X^{\top} D(\lambda u + (1 - \lambda) v) X)^{-1} x \\ 
&\leq x^{\top} (\lambda X^{\top} D(u)X + (1 - \lambda) X^{\top}D(v) X)^{-1} x \\
&\leq \lambda x^{\top} (X^{\top} D(u) X)^{-1} x + (1 - \lambda) x^{\top} (X^{\top} D(v) X)^{-1} x \\ 
&= \lambda f(u) + (1 - \lambda) f(v).
\end{align*} 
In the last inequality, we use the convexity of the matrix inverse for positive definite matrices: 
$$(\lambda A + (1 - \lambda)B)^{-1} \preceq \lambda A^{-1} + (1 - \lambda) B^{-1}.$$ 
$\widetilde{V}_{\plugin}(e_Z)$ is thus convex in the vector $e_Z$. \hfill$\blacksquare$

\subsection{Proof of Theorem \ref{thm:Vopt}}\label{ssec:proofthmvopt}
Note that we present the proof of Theorem \ref{thm:Vopt} before the proof of Theorem \ref{thm:gamma-consistency} because the latter will use some results derived in the former. 

In light of the approximations in Figure \ref{fig:var-approx}, we begin with two lemmas showing that 
\begin{compactenum}[\qquad 1)]
\item $\underset{e_Z \in \mathcal{C}_n}{\sup} \text{ } |V_{\oracle}(e_Z) - \widetilde V_{\oracle}(e_Z)| \to 0$, \quad\text{and}
\item $\underset{e_Z \in \mathcal{C}_n}{\sup} \text{ } |\widetilde V_{\oracle}(e_Z) - \widetilde{V}_{\plugin}(e_Z)| \overset{p}{\to} 0$. 
\end{compactenum}
Lemma~\ref{lemma:firstapprox} shows that moving the inverse out of the
expectation to get a more tractable expression makes an asymptotically negligible difference. Lemma~\ref{lemma:secondapprox} then shows that plugging in sample estimates makes an asymptotically negligible difference to this more tractable variance expression.

\begin{lemma} \label{lemma:firstapprox}
 Under the conditions of Theorem~\ref{thm:Vopt},   
 $$\underset{e_Z \in \mathcal{C}_n}{\sup} \text{ } |V_{\oracle}(e_Z) - \widetilde V_{\oracle}(e_Z)| =O(n^{-1/2})$$
as $n\to\infty$, where the implied constant depends only
on $\eta$, $c$ and $C$.
\end{lemma}
\begin{proof} Recall that $V_{\oracle}(e_Z)$ is the conditional variance of $\sqrt{n}\bar{X}_{\co}^\top\hat\gamma_{\oracle}$ given $X_1,\dots,X_n$. 
Fix $e_Z \in \mathcal{C}_n$, and define $e_W(X) = p_{\at}(X) + p_{\co}(X) e_Z(X)$. Let $D = \diag(W_i - e_W(X_i))$. 
Then $W_i\simind\dbern(e_W(X_i))$ and we will show that
$$n\bigl|\bar{X}_{\co}^{\top} \bbE[(X^{\top}D^2X)^{-1}]\bar{X}_{\co} - \bar{X}_{\co}^{\top} \bbE[X^{\top}D^2X]^{-1}\bar{X}_{\co}\bigr| = O(n^{-1/2})$$
uniformly for any vector of nudge propensities $e_Z \in \mathcal{C}_n$. 

By the Cauchy-Schwarz inequality and the assumption that $\Vert X_i\Vert_2^2 \leq C$,  
\begin{multline}\label{eq:nowitsanorm}
n \bigl|\,\bar{X}_{\co}^{\top} \bbE[(X^{\top}D^2X)^{-1}]\bar{X}_{\co} - \bar{X}_{\co}^{\top} \bbE[X^{\top}D^2X]^{-1}\bar{X}_{\co}\bigr| \\
\leq Cn \big\Vert \bbE[(X^{\top}D^2X)^{-1}] - \bbE[X^{\top}D^2X]^{-1}\big\Vert_{\op},
\end{multline}
where $\Vert\cdot\Vert_{\op}$ denotes the operator norm of a matrix.
We will show that the right-hand side of~\eqref{eq:nowitsanorm} converges to zero. To prove this, we proceed in steps. First, write 
$$X^{\top}D^2X = \sum_{i = 1}^{n} (W_i - \hat{e}_W(X_i))^2 X_i X_i^{\top}.$$
Define the random matrices 
$$M_i = (W_i - e_W(X_i))^2X_iX_i^{\top} - e_W(X_i)(1 - e_W(X_i))X_i X_i^{\top},$$ 
which each have mean zero and are independent. Moreover, 
$$\Vert M_i^2\Vert_{\op} \leq \Vert(X_i X_i^{\top})^2\Vert_{\op} \leq C^2$$
by assumption that $\Vert X_i\Vert_2^2 \leq C$. By the matrix Hoeffding inequality (see \cite{Tropp-matrix}), for any $t>0$:
\begin{align} \label{eq:matrixhoeffding}
\bbP\bigl(\Vert X^{\top}D^2X - \bbE[X^{\top}D^2X]\Vert_{\op} \geq t \bigr) &\leq \bbP\left(\left\{\lambda_{\max}\left(\sum_{i = 1}^{n} M_i \right) \geq t\right\}\right) \notag\\
&\quad + \bbP\left(\left\{\lambda_{\min}\left(\sum_{i = 1}^{n} M_i \right) \leq -t\right\}\right) \notag\\ 
&\leq 2d e^{-t^2/8nC^2}.
\end{align} 
Now let $M = X^{\top}D^2X$. Then 
\begin{align*} 
\Vert n\bbE[M^{-1}] - n\bbE[M]^{-1}\Vert_{\op} 
&= n\bigl\Vert\bbE[M^{-1} - \bbE[M]^{-1}]\bigr\Vert_{\op} \\ 
&\leq n\bbE\bigl[\Vert M^{-1} - \bbE[M]^{-1}\Vert_{\op}\bigr] \\ 
&=n\bbE\bigl[\bigl\Vert M^{-1}\bigl(\bbE[M]-M\bigr)\bbE[M]^{-1} \bigr\Vert_{\op}\bigr]\\
&\leq n\bbE\bigl[\Vert M^{-1}\Vert_{\op} \times  \Vert M - \bbE[M]\Vert_{\op}\bigr]\times\bigl\Vert\bbE[M]^{-1}\bigr\Vert_{\op}. \end{align*} 
Now, we can use our assumptions (3) and (4) to see that, for $n$ sufficiently large,
\begin{align*} 
\bbE[M] &= \sum_{i = 1}^{n} e_W(X_i)(1 - e_W(X_i))X_i X_i^{\top} \\ 
&\succeq \sum_{i = 1}^{n} \eta(1 - \eta) X_i X_i^{\top} \\ 
&\succeq n \eta(1 - \eta) cI_d. \end{align*} 
Hence, $\Vert\bbE[M]^{-1}\Vert_{\op} = O(n^{-1})$. Similarly, for $n$ sufficiently large,
\begin{align*} 
M &= \sum_{i = 1}^{n} \bigl(W - e_W(X_i)\bigr)^2 X_i X_i^{\top} 
\succeq \sum_{i = 1}^{n} \eta^2 X_i X_i^{\top} 
\succeq n \eta^2 cI_d 
\end{align*} 
and so $\Vert M^{-1}\Vert_{\op} = O(n^{-1})$. As for the last term, we can use the matrix Hoeffding result~\eqref{eq:matrixhoeffding} to get 
\begin{align*} 
\bbE\bigl[\Vert M - \bbE[M]\Vert_{\op}\bigr] &= \int_0^{\infty} \bbP\bigl(\Vert M - \bbE[M]\Vert_{\op} \geq t\bigr) \rd t\\ 
&\leq 2d \int_0^{\infty} e^{-t^2/8nC^2} \rd t \\ 
&= d\sqrt{8\pi nC^2}. \end{align*} 
In summary, 
\begin{align*} |V_{\oracle}(e_Z) - \widetilde V_{\oracle}(e_Z)| &\leq nC\underbrace{\bigl\Vert\bbE[M]^{-1}\bigr\Vert_{\op}}_{O(n^{-1})} \bbE\bigl[\underbrace{\Vert M^{-1}\Vert_{\op}}_{O(n^{-1})} \Vert M - \bbE[M]\Vert_{\op}\bigr] \\ 
&=O(n^{-1}) \underbrace{\bbE\bigl[\Vert M - \bbE[M]\Vert_{\op}\bigr]}_{O(n^{1/2})} \\ 
&=O(n^{-1/2}). \end{align*} 
Thus, we have an upper bound of order $O(n^{-1/2})$ that does not depend on our choice of $e_Z \in \mathcal{C}_n$, so the limit is indeed zero uniformly. 
\end{proof}
\begin{lemma} \label{lemma:secondapprox}
Under the conditions of Theorem~\ref{thm:Vopt},
$$\underset{e_Z \in \mathcal{C}_n}{\sup} \text{ } |\widetilde V_{\oracle}(e_Z) - \widetilde{V}_{\plugin}(e_Z)| \overset{p}{\to} 0$$
as $n\to\infty$.
\end{lemma}
\begin{proof}
Let $E = \diag(e_W(X)(1 - e_W(X))$ and $\hat{E} = \diag(\hat{e}_W(X)(1 - \hat{e}_W(X))$. As in Lemma \ref{lemma:firstapprox}, it suffices to show that 
$$n\Vert(X^{\top}EX)^{-1} - (X^{\top}\hat{E}X)^{-1}\Vert_{\op} \overset{p}{\to} 0,$$
Write $Q = X^{\top}EX$ and $\hat{Q} = X^{\top}\hat{E}X$, in which case 
\begin{align*} 
n\Vert Q^{-1} - \hat{Q}^{-1}\Vert_{\op} &
= n\Vert Q^{-1}(\hat{Q} - Q)\hat{Q}^{-1}\Vert_{\op} \\ 
&\leq n\Vert Q^{-1}\Vert_{\op} \Vert\hat{Q}^{-1}\Vert_{\op} \Vert Q - \hat{Q}\Vert_{\op}.
\end{align*} 
The first operator norm term is $O(n^{-1})$ as shown in Lemma \ref{lemma:firstapprox}. For the second, recall the assumption that $\Vert\hat{e}_W(X) - e_W(X)\Vert_{\infty}$ is $o_p(1)$ as $n\to\infty$. With probability tending to one, therefore, $\hat{e}_W(X)(1 - \hat{e}_W(X))$ is within $\eta(1 - \eta)/2$ of the true $e_W(X)(1 - e_W(X))$ for all $X$. Hence, with probability tending to one, 
\begin{align*} 
\hat{Q} &= \sum_{i = 1}^{n} \hat{e}_W(X_i)(1 - \hat{e}_W(X_i)) X_i X_i^{\top} \\ 
&\succeq \sum_{i = 1}^{n} \frac{1}{2} \eta(1 - \eta) X_i X_i^{\top} \\ 
&\succeq \frac{n}{2} \eta(1 - \eta) c I_d, \end{align*} 
meaning that $\Vert\hat{Q}^{-1}\Vert_{\op} = O_p(n^{-1})$. 

Finally, to address the last term, we have
\begin{align*} 
\Vert Q - \hat{Q}\Vert_{\op} &= \Vert X^{\top}(E - \hat{E})X\Vert_{\op} \\ 
&= \underset{\Vert v\Vert_2 = 1}{\sup} \sum_{i = 1}^{n} (X_i^{\top} v)^2 (E_{ii} - \hat{E}_{ii})\\ 
&\leq \underset{\Vert v\Vert_2 = 1}{\sup} \sqrt{\sum_{i = 1}^{n} (X_i^{\top}v)^4} \sqrt{\sum_{i = 1}^{n} (E_{ii} - \hat{E}_{ii})^2}.
\end{align*}
Now $\sqrt{\sum_{i=1}^n (X_i^\top v)^4}=O(n^{1/2})$ because $\Vert X_i\Vert_2\le C$, 
and
\begin{align*}
 \sum_{i = 1}^{n} (E_{ii} - \hat{E}_{ii})^2
&=\sum_{i = 1}^{n} \bigl(e_W(X_i) - \hat{e}_W(X_i)\bigr)^2\bigl(1 + e_W(X_i) + \hat{e}_W(X_i)\bigr)^2\\ 
&\leq 9\sum_{i = 1}^{n} (e_W(X_i) - \hat{e}_W(X_i))^2\\
&=o_p(n)
\end{align*} 
using $1 + e_W(X_i) + \hat{e}_W(X_i) \leq 3$ and 
$\Vert\hat{e}_W - e_W\Vert_{\ell^{\infty}} = o_p(1)$. As a result, $\Vert Q-\hat Q\Vert_{\op}=o_p(n)$, and so
$$n\Vert Q^{-1} - \hat{Q}^{-1}\Vert_{\op} \leq n\underbrace{\Vert Q^{-1}\Vert_{\op}}_{O(n^{-1})} \underbrace{\Vert\hat{Q}^{-1}\Vert_{\op}}_{O_p(n^{-1})} \underbrace{\Vert Q - \hat{Q}\Vert_{\op}}_{o_p(n)} \overset{p}{\to} 0.$$
\end{proof}

To conclude the proof of Theorem~\ref{thm:Vopt}, observe that $V_\oracle(\hat{e}_Z^*) \geq V_\oracle(e_Z^*)$ by definition, so it suffices to show that the difference is asymptotically nonpositive. We have 
\begin{align*} V_\oracle(\hat{e}_Z^*) - V_\oracle(e_Z^*) &= V_\oracle(\hat{e}_Z^*) - \widetilde V_{\plugin}(\hat{e}_Z^*) \\
&\quad + \widetilde V_{\plugin}(\hat{e}_Z^*) - \widetilde V_{\plugin}(e_Z^*)\\
&\quad + \widetilde V_{\plugin}(e_Z^*) - V_\oracle(e_Z^*). \end{align*}
Lemmas \ref{lemma:firstapprox} and \ref{lemma:secondapprox} imply that $\underset{e_Z \in \mathcal{C}_n}{\sup} \text{ } |V_{\oracle}(e_Z) - \widetilde{V}_{\plugin}(e_Z)| \overset{p}{\to} 0$ by the triangle inequality, so the first and third terms on the right-hand side converge in probability to zero. The middle term, meanwhile, is always nonpositive by definition of $\hat{e}_Z^*$ as the minimizer of $\widetilde V_{\plugin}$. Hence, we have the desired result. \hfill$\blacksquare$

\subsection{Proof of Theorem \ref{thm:gamma-consistency}}\label{ssec:proofthmgammaconsistency}
To show that $||\hat{\gamma}_{\oracle} - \hat{\gamma}_{\plugin}||_2 \overset{p}{\to} 0$, write 
\begin{align*} \hat{\gamma}_{\oracle} &= (X^{\top}D^2X)^{-1} X^{\top}D(Y - m^*),\quad\text{and} \\
\quad \hat{\gamma}_{\plugin} &= (X^{\top}\hat{D}^2X)^{-1}X^{\top}\hat{D}(Y - \hat{m}^*).
\end{align*}
Here, $D = \diag(W_i - e_W(X_i))$ and $\hat{D} = \diag(W_i - \hat{e}_W(X_i))$, while $m^* = m^*(X, Z, W)$ and $\hat{m}^* = \hat{m}^*(X, Z, W)$ are $n$-dimensional vectors with entries as defined in \eqref{eq: m*} and \eqref{eq:m*hat}, respectively. Define $M = X^{\top}D^2X$ and $\hat{M} = X^{\top}\hat{D}^2X$. Then 
\begin{align*} 
\Vert\hat{\gamma}_{\oracle} - \hat{\gamma}_{\plugin}\Vert_2 &= \Vert M^{-1}X^{\top}D(Y - m^*) - \hat{M}^{-1}X^{\top}\hat{D}(Y - \hat{m}^*)\Vert_2 \\ 
&\leq \Vert M^{-1}X^{\top}D(Y - m^*) - M^{-1}X^{\top}D(Y - \hat{m}^*)\Vert_2 \\ 
&\phe + \Vert M^{-1}X^{\top}D(Y - \hat{m}^*) - M^{-1}X^{\top}\hat{D}(Y - \hat{m}^*)\Vert_2 \\
&\phe+ \Vert M^{-1}X^{\top}\hat{D}(Y - \hat{m}^*) - \hat{M}^{-1}X^{\top}\hat{D}(Y - \hat{m}^*)\Vert_2 \\ 
&:= T_1 + T_2 + T_3.
\end{align*} 
It suffices to show that each of these three terms converges in probability to zero. 

For term $T_1$, we have
\begin{align*}T_1&= \Vert M^{-1} X^{\top}D(\hat{m}^* - m^*)\Vert_2 \leq \underbrace{\Vert M^{-1}\Vert_{\op}}_{O(n^{-1})} \underbrace{\Vert X^{\top}D\Vert_{\op}}_{O(n^{1/2})} \underbrace{\Vert\hat{m}^* - m^*\Vert_2}_{o_p(n^{1/2})}\\
&=o_p(1).
\end{align*}
Here, we use results derived in the proof of Theorem \ref{thm:Vopt} to bound the first two terms and use the assumption that $\Vert\hat{m}^* - m^*\Vert_2 = o_p(n^{1/2})$ to bound the last term.  

For term $T_2$, we have 
\begin{align*}T_2 &= \Vert M^{-1} X^{\top}(D - \hat{D})(Y - \hat{m}^*)\Vert_2 \\
&\leq \Vert M^{-1}\Vert_{\op} \Vert X^{\top}\Vert_{\op} \Vert D - \hat{D}\Vert_{\op} \Vert Y - \hat{m}^*\Vert_2.
\end{align*}
Define $\tilde{\varepsilon}(X, Z, W) = \varepsilon - \mathbb{E}[\varepsilon \st X, Z, W]$. We have 
\begin{align*} 
\Vert Y - \hat{m}^*\Vert_2 &\leq \Vert Y - m^*\Vert_2 + \Vert m^* - \hat{m}^*\Vert_2 \\ 
&= \Vert(W - e_W(X))X\gamma + \tilde{\varepsilon}\Vert_2 + \Vert m^* - \hat{m}^*\Vert_2 \\
&\leq \underbrace{\Vert X\gamma\Vert_2}_{O(n^{1/2})} + \underbrace{\Vert\tilde{\varepsilon}\Vert_2}_{O_p(n^{1/2})} + \underbrace{\Vert m^* - \hat{m}^*\Vert_2}_{o_p(n^{1/2})}.  
\end{align*} 
These bounds follow respectively from the assumption that $\Vert X_i\Vert_2^2 \leq C$, the assumption that the $\tilde{\varepsilon}$ terms are independent and homoscedastic, and the assumption that $\Vert m^* - \hat{m}^*\Vert_2 = o_p(n^{1/2})$. This sum is therefore $O_p(n^{1/2})$. Then 
$$T_2 \leq \underbrace{\Vert M^{-1}\Vert_{\op}}_{O(n^{-1})} \underbrace{\Vert X^{\top}\Vert_{\op}}_{O(n^{1/2})} \underbrace{\Vert D - \hat{D}\Vert_{\op}}_{o_p(1)} \underbrace{\Vert Y - \hat{m}^*\Vert_2}_{O_p(n^{1/2})} = o_p(1).$$

Finally, 
\begin{align*}T_3 &= \Vert(M^{-1} - \hat{M}^{-1})X^{\top}\hat{D}(Y - \hat{m}^*)\Vert\\
&\leq \Vert M^{-1} - \hat{M}^{-1}\Vert_{\op} \Vert X^{\top}\hat{D}\Vert_{\op} \Vert Y - \hat{m}^*\Vert_2.
\end{align*}
We have shown that the latter two terms are $O(n^{1/2})$ and $O_p(n^{1/2})$ already, so it suffices to show that the first term is $o_p(n^{-1})$. Indeed, 
\begin{align*}\Vert M^{-1} - \hat{M}^{-1}\Vert_{\op} &\leq \Vert M^{-1}\Vert_{\op} \Vert\hat{M}^{-1}\Vert_{\op} \Vert M - \hat{M}^{-1}\Vert_{\op} \\ 
&= \Vert M^{-1}\Vert_{\op} \Vert\hat{M}^{-1}\Vert_{\op} \Vert X^{\top}(D - \hat{D})X\Vert_{\op} \\ 
&\leq \underbrace{\Vert M^{-1}\Vert_{\op}}_{O(n^{-1})} \underbrace{\Vert\hat{M}^{-1}\Vert_{\op}}_{O_p(n^{-1})} \underbrace{\Vert X\Vert_{\op}^2}_{O(n)} \underbrace{\Vert D - \hat{D}\Vert_{\op}}_{o_p(1)}. 
\end{align*} 
This product is thus $o_p(n^{-1})$. Then $T_3$ is $o_p(1)$, which completes the proof. 
\hfill$\blacksquare$

\subsection{Proof of Theorem \ref{thm:constraintresults}}\label{ssec:proofconstraintresults}
For part (a) of the theorem, write
$$\widetilde{V}_{\plugin}(e_Z) = \bar{X}_{\co}^{\top} \bigl(X^{\top}\diag(\hat{e}_W(X)(1 - \hat{e}_W(X)))X\bigr)^{-1}\bar{X}_{\co},$$
where $\hat{e}_W(X) = \hat{p}_{\at}(X) + \hat{p}_{\co}(X) e_Z(X)$. For candidate nudge propensity vectors $\hat e_{Z, 1}$ and $\hat e_{Z, 2}$, let $\hat e_{W, 1}$ and $\hat e_{W, 2}$ be the corresponding treatment propensity vectors. Then 
$$\hat{e}_{W, 1}(1 - \hat{e}_{W, 1}) \succeq \hat{e}_{W, 2}(1 - \hat{e}_{W, 2}) \implies \widetilde{V}_{\plugin}(\hat e_{Z, 1}) \leq \widetilde{V}_{\plugin}(\hat e_{Z, 2}).$$
Since $x(1 - x)$ is maximized at $x = 1/2$, the optimal solution if possible is to take 
\begin{equation} \label{eq:eW_opt_RCT} \hat{e}_W(X_i) = \frac{1}{2} \iff e_Z(X_i) = \frac{1 - 2\hat{p}_{\at}(X_i)}{2\hat{p}_{\co}(X_i)}. \end{equation} 
If this is not in $[0, 1]$, then the optimal solution is to take either $e_Z(X_i) = 0$ or $e_Z(X_i) = 1$, depending on which yields $\hat{e}_W(X_i)$ closer to $1/2$. From \eqref{eq:eW_opt_RCT}, we see that this will be $0$ if $\hat{p}_{\at} > 1/2$ and $1$ if $\hat{p}_{\at} + \hat{p}_{\co} < 1/2$. In summary, 
$$\hat{e}_Z^*(X_i) = \begin{cases}\frac{1 - 2\hat{p}_{\at}(X_i)}{2\hat{p}_{\co}(X_i)}, & p_{\at} \leq \frac{1}{2} \text{ and } p_{\at} + p_{\co} \geq \frac{1}{2} \\ 
0, & p_{\at} > \frac{1}{2} \\ 
1, & p_{\at} + p_{\co} < \frac{1}{2}.
\end{cases} $$

For part (b), consider the case where we have a budget constraint, i.e., $\mathcal{C}_n = \{e_Z \in [0, 1]^n \st \bar{\hat{e}}_W = \mu\}$. Taking $e = e_Z$, the Lagrangian is 
\begin{align*}
L(e; \lambda, \tilde{\lambda}, \nu) &= \bar{X}_{\co}^{\top} \bigl(X^{\top}\diag(\hat{e}_W(X)(1 - \hat{e}_W(X)))X\bigr)^{-1} \bar{X}_{\co} \\ 
&+ \sum_{i = 1}^{n} \lambda_i(e_i - 1) - \tilde\lambda_ie_i + \nu\bigl((\hat{p}_{\at}(X_i) + \hat{p}_{\co}(X_i)e_i)^{\top}1 - \mu\bigr),
\end{align*}
where $\lambda_i$ and $\tilde\lambda_i$ are Lagrange multipliers for the constraints $e_i\le1$ and $e_i\ge0$, while $\nu$ is the Lagrange multiplier for the budget constraint.
Let $\hat{D}(e) = \diag(\hat{e}_W(X)(1 - \hat{e}_W(X))$ and $\hat{M}(e) = X^{\top}\hat{D}(e)X$, in which case
\begin{align*} \frac{\partial}{\partial e_i} \bar{X}_{\co}^{\top} \hat{M}(e)^{-1} \bar{X}_{\co} &= \bar{X}_{\co}^{\top} \left(\frac{\partial}{\partial e_i} \hat{M}(e)^{-1}\right) \bar{X}_{\co} \\ 
&= -\bar{X}_{\co}^{\top}\hat{M}(e)^{-1} \left(\frac{\partial}{\partial e_i} \hat{M}(e)\right) \hat{M}(e)^{-1} \bar{X}_{\co} \\ 
&=-\bar{X}_{\co}^{\top}\hat{M}(e)^{-1} X^{\top} \hat{D}_i(e) X \hat{M}(e)^{-1} \bar{X}_{\co},
\end{align*} 
where
$$\hat{D}_i(e) = \frac{\partial}{\partial e_i} \hat{D}(e) = \begin{bmatrix} 0 \\ & \ddots \\ & & \hat{p}_{\co}(X_i)(1 - 2 \hat{e}_W(X_i)) \\ & & & \ddots \\ & & & & 0\end{bmatrix}.$$
The stationarity condition for the Lagrangian is then 
$$\frac{\partial L}{\partial e_i} = -\bar{X}_{\co}^{\top} \hat{M}(e)^{-1}X^{\top} \hat{D}_i(e)X \hat{M}(e)^{-1} \bar{X}_{\co} + \lambda_i - \tilde{\lambda}_i - \nu = 
0 \text{ for all } 1 \leq i \leq n.$$
To prove that $\hat{e}_Z^*$ should be constant, first note that this would imply that $\lambda_i = \tilde\lambda_i = 0$ for all $i$ by complementary slackness. This would then reduce to 
$$\bar{X}_{\co}^{\top}\hat{M}(e)^{-1} X^{\top} \hat{D}_i(e) X \hat{M}(e)^{-1} \bar{X}_{\co} = \nu \text{ for all } 1 \leq i \leq n.$$
Thus, it suffices to show that, for $\hat{e}_Z$ chosen such that $\hat{e}_W = \mu 1_n$ is a constant vector, the gradient of the objective is the same in each entry. At this choice, $D(\hat{e}_Z) = \mu(1 - \mu)I_n$ is a multiple of the identity. Moreover, note that 
$$\bar{X}_{\co} = \frac{1}{n} \sum_{i = 1}^{n} X_i \hat{p}_{\co}(X_i) \propto X^{\top} 1_n,$$
where here we use that $\hat{p}_{\co}(X_i)$ is constant in $i$ by assumption. Then 
\begin{align*} 
\bar{X}_{\co}^{\top}\hat{M}(e)^{-1} X^{\top} \hat{D}_i(e) X \hat{M}(e)^{-1} \bar{X}_{\co} &\propto 1_n^{\top} X(X^{\top}X)^{-1}X^{\top} \hat{D}_i(e) X(X^{\top}X)^{-1}X^{\top}1_n \\ 
&= 1_n^{\top}H \hat{D}_i(e) H 1_n \\ 
&= 1_n^{\top} \hat{D}_i(e) 1_n \\ 
&= \hat{p}_{\co}(X_i)(1 - 2\mu). 
\end{align*} 
Here, $H = X(X^{\top}X)^{-1}X^{\top}$ is the standard hat matrix, and we use that $H1_n = 1_n$ since $1_n \in \mathrm{span}(X)$ by assumption. Moreover, since the compliance probability $\hat{p}_{\co}$ is constant by assumption, this quantity is constant in $i$, and so the KKT conditions are satisfied. Consequently, $\hat{e}_Z^* = (\mu - \bar{\hat{p}}_{\at})/{\hat{p}_{\co}}$ is a solution to the constrained optimization.

\section{Cross-fitting for estimation}\label{sec:cross-fitting}

In this section, we outline a cross-fitting procedure to estimate $\gamma$ at a $\sqrt{n}$ rate. This procedure will require stronger assumptions than those of Theorem \ref{thm:gamma-consistency} and is more complicated, so we omitted it from the main text for simplicity. The method resembles that of \cite{robinson}, and the proof is essentially that of Theorem 4.2 of \cite{wagerbook}. 

Suppose that we have run the entire design procedure and have data $\{X_i, Z_i, W_i, Y_i\}_{i \leq n}$ from the main study (we do not use the pilot study in what follows because it has been used to pick $e_Z$, which would introduce dependencies). To produce the estimate $\hat{\gamma}_{\plugin}$ described in Section \ref{ssec:LATEest}, we need to form estimates $\hat{e}_W(X)$ and $\hat{m}^*(X, Z, W)$ and then run the regression $Y - \hat{m}^*(X, Z, W) \sim (W - \hat{e}_W(X))X$. 

However, dependencies can arise if one uses the same data to estimate nuisance functions and to fit the regression. Thus, a standard approach in semiparametric causal inference is to use cross-fitting, in which the data are divided into $K > 1$ folds to ensure that fitting and estimation are independent. Specifically, partition the data set into subsets of (approximately) equal size $\mathcal{D}_1, \mathcal{D}_2, \ldots, \mathcal{D}_K$ and, for fold $k$, estimate $\hat{e}_W(X)^{(k)}$ and $\hat{m}^*(X, Z, W)^{(k)}$ within each dataset. Define
\begin{align*} 
\hat{e}_W(X)^{(-k)} &= \frac{1}{K - 1} \sum_{i \neq k} \hat{e}_W(X)^{(i)} \\ 
\hat{m}^*(X, Z, W)^{(-k)} &= \frac{1}{K - 1} \sum_{i \neq k} \hat{m}^*(X, Z, W)^{(i)}
\end{align*}
to be the average of the estimates on all other folds besides fold $k$. To obtain $\hat{\gamma}_{\plugin}^{(k)}$, regress $Y - \hat{m}^*(X, Z, W)^{(-k)}$ on $(W - \hat{e}_W(X)^{(-k)})X$. Finally, define
$$\hat{\gamma}_{\plugin} = \frac{1}{K} \sum_{k = 1}^{K} \hat{\gamma}_{\plugin}^{(k)}.$$
For data point $i$, we let $k(i)$ denote the fold in which it is placed. The following theorem establishes that this procedure estimates $\gamma$ at a $\sqrt{n}$ rate under suitable convergence assumptions on estimation of $\hat{m}^*(X, Z, W)^{(k)}$ and $\hat{e}_W(X)^{(k)}$ in each fold. We emphasize that these rates are stronger than those in Theorem \ref{thm:gamma-consistency}, which only proves consistency of $\hat{\gamma}_{\plugin}$.  
\begin{theorem}\label{thm:cross-fitting}
Suppose that the following conditions all hold: 
\begin{enumerate} 
\item The propensity scores are estimated sufficiently well:
$$\Vert\hat{e}_W(X)^{(-k)} - e_W(X)\Vert^2_{2} = \sum_{i: \text{ } k(i) = k} (\hat{e}_W(X_i)^{-(k)} - e_W(X_i))^2 = o_p(n^{3/2})$$.
\item The $m^*(X, Z, W)$ functionals are estimated sufficiently well: 
$$\Vert\hat{m}^*(X, Z, W)^{(-k)} - m^*(X, Z, W)\Vert^2_{2} = \sum_{i: \text{ } k(i) = k} (\hat{m}^*(X_i, Z_i, W_i)^{(-k)} - m^*(X_i, Z_i, W_i))^2 = o_p(n^{3/2})$$.
\item The rows of $X$ are bounded, i.e., $||X_i||_{\infty} \leq M$ for some $M < \infty$.
\item There exists $\eta > 0$ such that $e_W(X) = p_{\at}(X) + p_{\co}(X) e_Z(X) \in [\eta, 1 - \eta]$ for all $X \in \mathcal{X}$. 
\end{enumerate} 
Write $\tilde{Y}_i = Y_i - m^*(X_i, Z_i, W_i)$ and $\tilde{X}_i = (W_i - e_W(X_i))X_i$, so that
$$\tilde{Y}_i = \tilde{X}_i^{\top} \gamma + \tilde{\varepsilon}_i$$
and $\bbE[\tilde{\varepsilon}_i \st X_i, Z_i, W_i] = 0.$ Then if $\Var(\tilde{X}_i)$ is full rank,
$$\sqrt{n}(\hat{\gamma}_{\plugin} - \gamma) \overset{d}{\to} N(0, V_{\gamma}),$$
where
$$V_{\gamma} = \Var(\tilde{X}_i)^{-1} \bbE[\tilde{\varepsilon}_i^2 \tilde{X}_i\tilde{X}_i^{\top}]\Var(\tilde{X}_i)^{-1}.$$
\end{theorem}
\begin{proof}
See Theorem 4.2 of \cite{wagerbook}. 
\end{proof}

Theorem \ref{thm:cross-fitting} gives the same CLT for $\hat{\gamma}_{\plugin}$ that we would have if we knew the true nuisance functions, and thus could use $\hat{\gamma}_{\oracle}$ directly.
The proof consists of using our assumptions to show that $\hat{\gamma}_{\plugin} - \hat{\gamma}_{\oracle} = o_p(n^{-1/2})$. Under the conditions of \ref{thm:cross-fitting}, $\hat{\gamma}_{\plugin}$ is also asymptotically linear. Hence, we could bootstrap the $\hat{\gamma}_{\plugin}$ estimator obtained via cross-fitting to obtain an asymptotically valid confidence interval for $\gamma$.

\section{Generalization to weighted least squares} \label{sec:WLS}
Suppose first that we allow $\Var(\varepsilon \st X, Z, W) = \sigma^2(X)$ to vary with $X$ (we explain subsequently the case of dependence on both $X$ and $W$). Recall from the formula \eqref{OLSeqn} that 
$$Y - m^*(X, Z, W) = (W - e_W(X))X\gamma + \varepsilon - \bbE[\varepsilon \st X, Z, W].$$
Under heteroscedasticity, we can additionally use the pilot study to estimate $\sigma^2(X)$. This can be done simply by noting that 
$$Y - m(X, Z, W) = \varepsilon - \bbE[\varepsilon \st X, Z, W],$$
where $m(X, Z, W) = \bbE[Y \st X, Z, W]$, and so we need only estimate the variance of $Y - m(X, Z, W)$ as a function of $X$. Letting $\hat{\sigma}^2(X)$ be such a function and $\hat{\Sigma}$ be the diagonal matrix whose entries are $\hat{\sigma}^2(X_i)$, this immediately suggests the weighted least squares estimator:
$$\hat{\gamma}_{\mathrm{WLS}} = (X^{\top} \hat{D} \hat{\Sigma}^{-1}\hat{D} X)^{-1} X^{\top}\hat{D} \hat{\Sigma}^{-1} Y.$$
Likewise, the corresponding objective function to minimize is simply
$$\widetilde{V}_{\plugin}(\hat e_Z) = \bar{X}_{\co}^{\top}\left(X^{\top} \text{diag}\left(\frac{\hat{e}_W(X)(1 - \hat{e}_W(X))}{\hat{\sigma}^2(X)} \right) X\right)^{-1}\bar{X}_{\co}.$$
If we allow $\sigma^2(X, W)$ to vary with both $X$ and $W$, a non-convex objective can result in some cases.  However the following theorem provides a condition under which convexity is preserved. 
\begin{theorem}\label{thm:hetero_ratio}
If $\hat{\sigma}^2(X, W)$ varies with both $X$ and $W$ and the condition 
$$\frac{1}{2} \leq \frac{\hat{\sigma}^2(X_i, 1)}{\hat{\sigma}^2(X_i, 0)} \leq 2$$
holds for all $X_i$ with $1 \leq i \leq n$, then the weighted least squares procedure yields a convex objective function $\widetilde{V}_{\plugin}(e_Z)$. 
\end{theorem}
\begin{proof} 
We focus our attention to the objective $\widetilde V_{\oracle}(e_Z)$ rather than $\widetilde{V}_{\plugin}(e_Z)$, since the latter merely substitutes in estimated functions for oracle ones. The final conclusion will then hold with $\sigma^2(X, W)$ replaced by $\hat{\sigma}^2(X, W)$. In addition, $\widetilde V_{\oracle}$ is convex in $e_Z(X)$ if and only if it is convex in $e_W(X)$ since the latter is a linear function of the former with positive coefficients, and so we can focus on convexity in $e_W(X)$. 

To see the possible non-convexity, recall from the derivation in Section \ref{ssec:design} that we must move a conditional expectation inside of a matrix inverse when switching from $V_{\oracle}(e_Z)$ to $\widetilde V_{\oracle}(e_Z)$, and that this expectation passes to the innermost diagonal matrix. In the fully heteroscedastic setting, the $i$'th term in this final expectation is
\begin{equation} \label{eq:hetero}\begin{split}  g_i(e_W(X_i)) &:= \bbE\left[\frac{(W_i - e_W(X_i))^2}{\sigma^2(X_i, W_i)} \st X_i\right] \\
&= e_W(X_i) \frac{(1 - e_W(X_i))^2}{\sigma^2(X_i, 1)} + (1 - e_W(X_i)) \frac{e_W(X_i)^2}{\sigma^2(X_i, 0)}. 
\end{split} 
\end{equation}
The variance  $\widetilde V_{\oracle}$ is then
\begin{align*} \widetilde V_{\oracle} &= \bar{X}_{\co}^{\top} (X^{\top} \text{diag}(g(e_W(X)) X)^{-1} \bar{X}_{\co} \\ 
&= \bar{X}_{\co}^{\top} \left( \sum_{k = 1}^{n} g_i(e_W(X_i)) X_i X_i^{\top}\right)^{-1} \bar{X}_{\co}. \end{align*}
By the vector composition rules for convex functions (e.g., Chapter 3 of \cite{boyd2004convex}), $\widetilde V_{\oracle}$ is convex in $e_W(X)$ if each $g_i$ function as defined in \eqref{eq:hetero} is concave. This cannot be guaranteed for general $\sigma^2(X_i, 1)$ and $\sigma^2(X_i, 0)$. However, there is still a wide range of allowable values of $\{\sigma^2(X, 1), \sigma^2(X, 0)\}$. Taking the second derivative of $g_i(e_W(X_i))$ with respect to $e_W(X_i)$, we arrive at
\begin{equation} \label{eq:hetero-second-deriv} \frac{1}{\sigma^2(X_i, 0)}\left(6 e_W(X_i) - 4\right) + \frac{1}{\sigma^2(X_i, 1)}\left(2 - 6e_W(X_i)\right).\end{equation}
We seek conditions on the variance ratio under which this is guaranteed to be non-positive for any $e_W(X_i) \in [0, 1]$. This is a linear function in $e_W(X_i)$, so it suffices to check non-positivity at the endpoints $e_Z = 0$ and $e_Z = 1$. Even more extreme, we can simply check non-positivity at the endpoints $e_W = 0$ and $e_W = 1$: 
\begin{align*} 
&e_W(X_i) = 0: \frac{- 4}{\sigma^2(X_i, 0)} + \frac{2}{\sigma^2(X_i, 1)} \leq 0 \iff \frac{\sigma^2(X_i, 1)}{\sigma^2(X_i, 0)} \geq \frac{1}{2}, \\ 
&e_W(X_i) = 1: \frac{2}{\sigma^2(X_i, 0)} - \frac{4}{\sigma^2(X_i, 1)} \leq 0 \iff \frac{\sigma^2(X_i, 1)}{\sigma^2(X_i, 0)} \leq 2.
\end{align*} 
Hence, an allowable interval for the variance ratio is $[1/2, 2]$. Moving from $\widetilde V_{\oracle}(e_Z)$ to $\widetilde{V}_{\plugin}(e_Z)$, the conclusions simply change to ones about $\hat{\sigma}^2(X_i, W_i)$, which completes the proof. 
\end{proof}
Note that the interval in Theorem \ref{thm:hetero_ratio} can be improved by instead checking the endpoints $e_Z(X_i) = 0$ and $e_Z(X_i) = 1$, which correspond to $e_W(X_i) = p_{\at}(X_i)$ and $e_W(X_i) = p_{\co}(X_i)$, respectively. These choices result in an interval for the variance ratio that is wider but that depends in a piecewise manner on the estimated compliance probabilities. Hence, we only state the simpler interval $[1/2, 2]$ here. 

In summary, if there is reason to believe that $\sigma^2(X, 1)$ and $\sigma^2(X, 0)$ differ for some values of $X$, our method can accommodate estimated variances whose ratio is not too extreme. Any deviations estimated to go beyond $[1/2, 2]$ can be regularized to lie in this interval, which may still be useful in the design stage. 
To suggest such a regularization, we first assume that $\min(\sigma^2(X,1),\sigma^2(X,0))>0$. We can then set the smaller variance to some value $\tau^2=\tau^2(X)$ and the larger one to $2\tau^2$. Because estimates are generally weighted in inverse proportion to variance, we can preserve the total importance $1/\sigma^2(X,0)+1/\sigma^2(X,1)$ 
of some value $X$ by solving
$$ \frac1{\tau^2}+\frac1{2\tau^2}=\frac1{\sigma^2(X,0)}+\frac1{\sigma^2(X,1)},$$
which yields
$$
\tau^2 = \frac{3/2}{1/\sigma^2(X,0)+1/\sigma^2(X,1)}.
$$

\section{Simulation details} \label{sec:expdetails}
All covariates were standardized to have mean zero and variance one at the outset to avoid a coefficient vector $\gamma$ with one or two substantially larger entries. In addition, because the risk score was so heavy-tailed, we opted to convert to one that is evenly spaced in $\{\frac{1}{n + 1},\dots, \frac{n}{n + 1}\}$ by taking the ordinal rank of a unit's risk score. We held out $20\%$ of the data in each simulation to be used in the pilot study. 

Given data $X$ and nudge propensity $e_Z(X)$, our data-generating procedure is as follows: 
\begin{align*} 
Z \st X &\sim e_Z(X), \\ 
\mathcal{C} \st X &\sim \text{Cat}(\{p_{\nt}(X), p_{\co}(X), p_{\at}(X)\}), \\
W \st \mathcal{C}, Z &= Z \indic\{\mathcal{C} = \co\} + \indic\{\mathcal{C} = \at\}, \\ 
Y \st X, W, C &\sim f(X) + WX\gamma + 10 \times \indic\{\mathcal{C} = \at\} - 10 \times \indic\{\mathcal{C} = \nt\} + N(0, 20).
\end{align*} 
Here, $\mathcal{C}$ is a unit's compliance status (either complier, always-taker, or never-taker). In addition, $f(X)$ and $\gamma$ are fit from a held-out subset of the data as described in Section \ref{sec:triage}. $f(X)$ is a random forest model, while $\gamma \in \bbR^7$ has the following values: 
\begin{align*} 
\text{(Intercept) } \gamma_0 &= -6.35 \\ 
\text{(Hours in ED) } \gamma_1 &= -32.31 \\ 
\text{(Weight in pounds) } \gamma_2 &= -10.65 \\ 
\text{(Age) } \gamma_3 &= 11.19 \\ 
\text{(Systolic blood pressure) } \gamma_4 &= -2.59 \\ 
\text{(Pulse) } \gamma_5 &= 62.03 \\ 
\text{(Risk score) } \gamma_6 &= -3.40 
\end{align*}
With this $\gamma$, the true final LATE is about $\tau_{\late} = -7.00$. The error term added to $Y$ depends on the unknown compliance status and is designed such that always-takers receive a boost in their length of stay and never-takers receive a decrease. Letting $r(X)$ be a patient's risk score, we chose the compliance probabilities as
\begin{align*} 
p_{\at}(X) &= 0.8 e^{-5r(X)} + 0.05 \\
p_{\at}(X) &= 0.8 e^{-5(1 - r(X))} + 0.05 \\
p_{\co}(X) &= 1 - p_{\at}(X) - p_{\nt}(X).
\end{align*}
These are visualized in Figure \ref{fig:comp-probs}. They incorporate an intuitively reasonable expectation that higher risk patients are much more likely to be always-takers and lower risk patients are much more likely to be never-takers. 

In addition, Figure \ref{fig:opt-objectives} compares the values of the objective function at each optimal constrained solution to provide a sense of how much each constraint set inflates the optimal (approximate) variance. The unconstrained (base) optimum gives an objective value of roughly $8.44 \times 10^{-4}$. For ease of interpretability, Figure \ref{fig:opt-objectives} plots the ratio of a given solution's objective value to this quantity. The monotonicity constraint only results in about a $6\%$ inflation of the objective, whereas using all three constraints with $\rho = 0.95$ results in about an $87\%$ inflation. 

\begin{figure}[t!] 
\centering
\includegraphics[width=12cm]{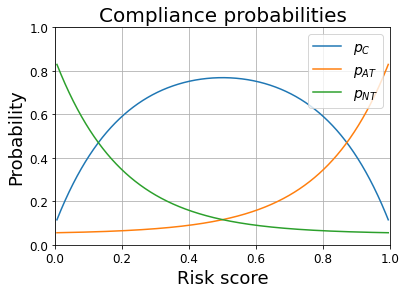}
\caption{Compliance probabilities as a function of transformed risk score. In this simulation, patients with very low risk are more likely to be never-takers, those with very high risk are more likely to be always-takers, and those in the middle are more likely to be compliers.}
\label{fig:comp-probs}
\end{figure}

\begin{figure}[t!] 
\centering
\includegraphics[width=13cm]{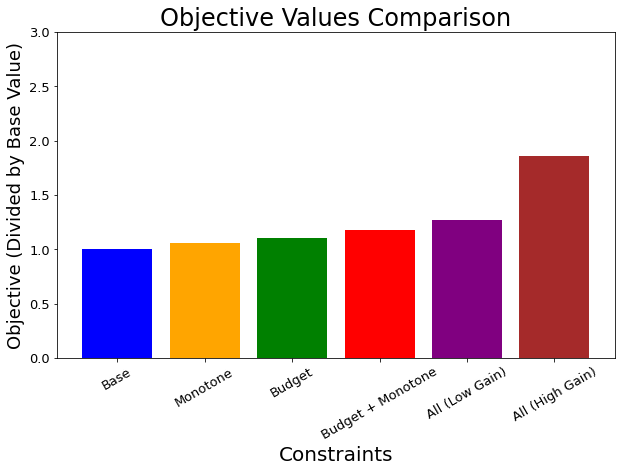}
\caption{Value of the objective function under each set of constraints, normalized by the value of the objective function at the optimal unconstrained solution.}
\label{fig:opt-objectives}
\end{figure}

\end{document}